\documentclass[10pt,journal]{IEEEtran}
\usepackage[utf8]{inputenc}
\usepackage[T1]{fontenc}
\usepackage{microtype} 
\usepackage{cite}
\usepackage{graphicx}
\usepackage[caption=false]{subfig}
\usepackage{amsfonts,amsmath,amsthm,amssymb}
\usepackage{algorithm} 
\usepackage{algpseudocode}
\usepackage{epstopdf}

\usepackage{amsmath} 

\usepackage{tikz}
\usetikzlibrary{automata,positioning}
\usetikzlibrary{decorations.pathreplacing,decorations.markings,shapes.geometric}
\usetikzlibrary{shapes,arrows}
\usetikzlibrary{backgrounds,calc,positioning}

\theoremstyle{definition}

\usepackage{mathtools}
\usepackage{commath}
\theoremstyle{remark}
\newtheorem{remark}{Remark}[section] 

\theoremstyle{plain}
\newtheorem{theorem}{Theorem}[section]

\newcommand{\BigO}[1]{\ensuremath{\operatorname{O}\bigl(#1\bigr)}}

\newcommand{\bmath}[1]{\mbox{\boldmath$#1$}}

\usepackage{cite}
\usepackage{hyperref}

\ifCLASSINFOpdf
\else
\fi
\begin{document}
%

\allowdisplaybreaks

\title{Redesigning Multipath TCP With Deep Q-Learning For Future Internet}
\title{Learning to Harness Bandwidth with Multipath Scheduling and Congestion Control}
\title{Learning to Harness Bandwidth with Multipath Congestion Control and Scheduling}
%
%
%
%
\author{Shiva~Raj~Pokhrel, \textit{Senior Member}, IEEE and Anwar~Walid, \textit{Fellow}, IEEE
\IEEEcompsocitemizethanks{\IEEEcompsocthanksitem S.~R.~Pokhrel is with the School of Information Technology, Deakin University, Geelong, Australia. A.~Walid is with Nokia Bell Labs, Murray Hill, NJ, USA. The revised version of this paper will appear in IEEE Transactions in Mobile Computing~\cite{pokhrelharness2021}.
Email: shiva.pokhrel@deakin.edu.au, anwar.walid@nokia-bell-labs.com
}
}

\IEEEtitleabstractindextext{%

\begin{abstract}
Multipath TCP (MPTCP) has emerged as a facilitator for {harnessing and pooling available bandwidth in wireless/wireline communication networks and in data centers}.  Existing implementations of MPTCP such as, Linked Increase  Algorithm (LIA), Opportunistic LIA (OLIA) and BAlanced LInked Adaptation (BALIA) include separate algorithms for congestion control and packet scheduling, with pre-selected control parameters. We propose a Deep Q-Learning (DQL) based framework for joint congestion control and packet scheduling for MPTCP. At the heart of the solution is an intelligent agent for interface, learning and actuation, which learns from experience optimal congestion control and scheduling mechanism using DQL techniques with policy gradients. We provide a rigorous stability analysis of system dynamics which provides important practical design insights. In addition, the proposed DQL-MPTCP algorithm utilizes the `recurrent neural network' and integrates it with `long short-term memory' for continuously i) learning dynamic behavior of subflows (paths) and ii) responding promptly to their behavior using prioritized experience replay. With extensive emulations, we show that the proposed DQL-based MPTCP algorithm outperforms MPTCP LIA, OLIA and BALIA algorithms. Moreover, the DQL-MPTCP algorithm is robust to time-varying network characteristics, and provides dynamic exploration and exploitation of paths. 
\end{abstract}

}
\maketitle
\IEEEdisplaynontitleabstractindextext

%
\IEEEpeerreviewmaketitle

\section{Introduction} 
\label{sec:introduction}

\label{sec:introduction}
\IEEEPARstart CONGESTION Control design of both TCP and multipath TCP (MPTCP) is a fundamental problem in networking and has been widely investigated in the literature (see ~\cite{walid2016balanced, li2019smartcc, goyal2019abc, chiariotti2019analysis, pokhrel2018improving, nie2019dynamic, xu2019experience, raj2021rent} and references therein). Most of the seminal single-path congestion control protocols proposed are either \emph{delay-based} (e.g., TCP Vegas), or \emph{loss-based} (e.g., TCP NewReno), otherwise hybrid (both) (e.g., Compound TCP). All of these congestion control protocols use some packet-based events (loss, delay, etc.) as an indication for the congestion and perform their window adjustment (based on some fixed control mechanism) to control the number of outstanding packets in the network. 

To take advantage of the increasingly available wireless and wireline access technologies, and mutlipath and multi-homing {capabilities in the Internet and data centers}, Multipath TCP (MPTCP) has emerged as an enabler for harnessing and aggregating  available bandwidth and improve applications' performance. Raiciu \emph{et al.}~\cite{raiciu2011coupled} developed the MPTCP \emph{Linked Increase Algorithm} (LIA) congestion control that elegantly uses their specific window increase function to couple the congestion windows running on different subflows; however, it suffers from throughput inefficiency problem which was addressed by \emph{Opportunistic} LIA (OLIA)~\cite{khalili2013mptcp}. Packet scheduling in MPTCP determines which path, among available paths with open windows, to send a packet so that the need for packet re-sequencing at the receiver is minimized. The design of MPTCP packet scheduling has also been equally challenging~\cite{hurtig2018low, garcia2017low,  lim2017ecf, 8666496}. 

\textcolor{black}{As shown in Fig.~\ref{fig:my2}, MPTCP has a vital role to play in 5G and 6G networks~\cite{TS23401_MPTCP}; see the 3rd Generation Partnership Project (3GPP) 5G Release 17~\cite{TS23401_MPTCP} for further details. Since Release 15, the 3GPP introduced a new  Access Traffic Steering, Switching, and Splitting (ATSSS) feature~\cite{simon2020atsc}.}\footnote{\textcolor{black}{{\textit{Steering}}: choosing the best network depending on the user’s location and network conditions;
{\textit{Switching}}:enabling seamless handovers from 4G/5G to Wi-Fi (or vice versa) 
{\textit{Splitting}}: data transport over multiple paths for higher speeds.}} 

\begin{figure}[t]
    \centering
    \includegraphics[scale=0.35]{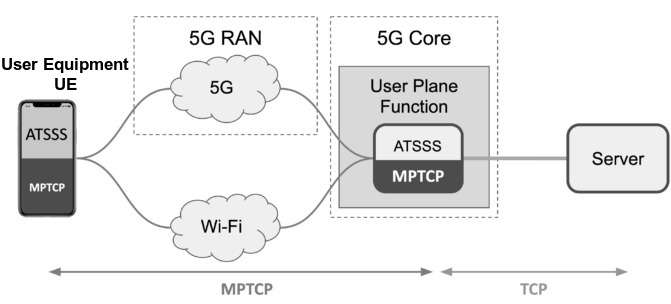}
    \caption{\textcolor{black}{5G Network-assisted MPTCP: An abstract view of ATSSS~\cite{TS23401_MPTCP}}}
    \label{fig:my2}
\end{figure}
\textcolor{black}{ATSSS is one of the prevalent use case of network-assisted multipath data transport. It has two main components (see Fig.~\ref{fig:my2}): User Equipment (UE) and 5G Core. The UE could natively use MPTCP while the User Plane Function may have a TCP to MPTCP proxy. The 3GPP has already defined the use of the 0-RTT TCP Convert Protocol for the MPTCP proxy, which addresses the proxy directly (IETF RFC-8803, July 2020)~\cite{bonaventure20170}. Another option relying on the activation of a dedicated function developed MPTCP \textit{concentrator}~\cite{boucadair2016mptcp}, where the identity is explicitly considered on the hosts.}

\textcolor{black}{One early very prevalent use of MPTCP is the SIRI digital assistant of Apple iPhone, which has been using MPTCP since 2013 (iOS 7) to allow Siri to use both the Wi-Fi or Cellular whatever is available and attain the lowest delay and high reliability. In 2019, MPTCP was also introduced by Apple in their Maps and Music. Samsung is another key player in the MPTCP framework that has been introduced in phones used by Korean Telecom, SK Telecom and LG+ since 2015 in South Korea. The goal here was to merge Wi-Fi and Cellular and use them to attain Gigabit speeds at the same time. From 2016, on customer premises devices such as DSL home gateways and/or LTE/4G modems, Tessares and Swisscom adopted MPTCP, mainly to aggregate the capacities of the fixed and mobile networks.}

By using fluid-flow model and dynamical system analysis, Peng \emph{et al.}~\cite{peng2016multipath} designed \emph{BAlanced LInked Adaptataion} (BALIA) to achieve a better balanced performance. However, all of the aforementioned congestion control approaches~\cite{peng2016multipath, pokhrel2018improving, raiciu2011coupled, khalili2013mptcp} designed MPTCP on the basis of \textit{model-based} frameworks, which are challenged in highly dynamic wireless networks (such as 5G\footnote{`Considerations for MPTCP operation in 5G' https://tools.ietf.org}~\cite{haile2020end} and beyond\footnote{`5G ATSSS solution', https://www.tessares.net}~\cite{lee2020deft, polese2017tcp,  pokhrellow}) as the underlying models do not apply equally well in various scenarios and the solutions can be in suboptimal regimes and unstable states.

To elaborate, the current MPTCP congestion control design mainly consists of two steps: i) developing mathematical models for resource allocation problem (e.g., fluid and/or queuing models~\cite{peng2016multipath, pokhrel2018improving}), ii) designing algorithms to solve the models (e.g., convex optimization based on its properties~\cite{peng2016multipath, pokhrel2018improving}). Such models, which are based on the Network Utility Maximization (NUM) or control-theoretic frameworks, typically lead to algorithms with fixed parameters. Such algorithms, however, may be unable to handle highly fluctuating and time-varying wireless channels and data paths.

Furthermore, in current model-based MPTCP design, a user's utility is normally a function of its instantaneous rate, yielding `bandwidth sharing' networks~\cite{pokhrel2019fair} which fits the case of infinite length flows. All the main single path TCP algorithms proposed in the literature have strictly concave utility functions, implying uniqueness and existence of a stable equilibrium and their convergence. However, the case of MPTCP is much more delicate: whether an underlying utility function exists depends on the design choice. It is worth noting that the utility function does not exist for some well-known MPTCP congestion control algorithms~\cite{peng2016multipath, pokhrel2018improving, raiciu2011coupled, khalili2013mptcp} (see~\cite[Tab.~1]{peng2016multipath}). Therefore, utility maximization with the standard MPTCP congestion control algorithms has been typically challenging and nontrivial. 

Moreover, the current MPTCP design does not provide applications with full access to the transport layer features (cross-layer)\footnote{Cross-layer approach refers to sharing information among TCP/IP layers for efficacious use of network resources and achieving high degree of adaptivity. It is an escape from the established TCP/IP layering of protocols.}, making it difficult to configure the transport layer behavior to take advantage of the available context and network information. Given that the file size is known, a natural alternative is to use, for a path, the optimal \emph{Shortest Remaining Processing Time} (SRPT) scheduling. However, naively applying SRPT to networks creates starvation for multipath flows. In contrast, one may plan to develop an optimization framework and/or cross-layer approach with i)  precise prediction of future values and ii)  accurate mathematical model to quantify the network behaviour precisely. But both i) and ii) are extremely challenging  and unscalable because of large state space (combinatorial in the number of packets in flights) over future wireless networks. Therefore, there is a need for robust experience-driven algorithms that can adapt intelligently to dynamic environments without escaping from the TCP/IP layering fundamentals.

\subsection{Literature}
Recently, real momentum has been building up towards learning-based TCP designs to exploit and handle dynamic network environments~\cite{dong2015pcc, xu2019experience, winstein2013tcp, zaki2015adaptive, li2019smartcc, mai2019self, liao2020precise, chung2017machine, silva2020avira, li2019hpcc, goyal2019abc}. Winstein \emph{et al.}~\cite{winstein2013tcp} proposed \emph{Remy}, a single path congestion control, which is capable of generating control rules for several different network settings. Authors in~\cite{dong2015pcc, dong2018pcc} developed a (single path) performance-based congestion control protocol, such that the source can continuously track and adopts actions timely for a higher network performance~\cite{lillicrap2015continuous}. Zaki \emph{et al.} \cite{zaki2015adaptive} proposed \emph{Verus}, a congestion control protocol that learns a correlation between end-to-end delivery delay of packets to update the congestion window~\cite{li2019hpcc,goyal2019abc}.  However, these works~\cite{winstein2013tcp, dong2015pcc, dong2018pcc, zaki2015adaptive, li2019hpcc, goyal2019abc}, are basically for designing single path congestion control protocols and are not applicable directly for MPTCP. 


More recently, advances in machine learning have been applied to the design of MPTCP~\cite{li2019smartcc, xu2019experience, mai2019self, liao2020precise, chung2017machine, silva2020avira, pokhrel2020multipath, 8666496}. Mai \textit{et al.}~\cite{mai2019self} applied deterministic policy gradient for learning the optimal congestion control strategies. Liao \textit{et al.}~\cite{liao2020precise} proposed a DRL-based multipath scheduler aiming at decreasing the out-of-order queue sizing under heterogeneous paths. Silva \textit{et al.}~\cite{silva2020avira} implemented an adaptive virtual reality with content-aware prioritisation to improve MPTCP's performance. 

Of particular relevance to this work are the two multipath congestion control designs~\cite{li2019smartcc, xu2019experience} using \textit{deep reinforcement learning} (DRL). Both of these designs~\cite{li2019smartcc, xu2019experience} focus on congestion control and do not address packet scheduling. Indeed, all of the aforementioned MPTCP designs~\cite{li2019smartcc, xu2019experience, mai2019self, liao2020precise, chung2017machine, silva2020avira, pokhrel2020multipath,8666496} considered either congestion control or packet scheduling separately (instead of jointly). It is known that the performance of MPTCP is significantly impacted by packet scheduling which decides what packet to send on which path among the available paths that have open window. Such a decision impacts the goodput (rate of packets delivered to the receiver in the right order) and the required resquencing delay. The advantage of the MPTCP controller in~\cite{xu2019experience}, however, is shown only when it is applied to a group of sources, rather than the usual and versatile case of a controller per source. In addition, the stability of the proposed algorithms~\cite{li2019smartcc, xu2019experience, mai2019self, liao2020precise, chung2017machine, silva2020avira} are not addressed, whereas stability is a  key desirable goal in the original MPTCP design philosophy~\cite{raiciu2011coupled}. Authors in~\cite{pokhrel2020multipath,8666496} introduced Deep Q Network (DQN) framework to enhance the MPTCP congestion control~\cite{pokhrel2020multipath} and packet scheduling~\cite{8666496} performance in asymmetric paths.


In this paper, we develop DQL-MPTCP, a Deep Q-Learning (DQL) based framework for joint congestion control and packet scheduling of MPTCP. At the heart of the solution is an intelligent agent, which learns from experience optimal congestion control and scheduling mechanism using DQL techniques with policy gradients. DQL-MPTCP is a situation-aware learning driven by reward maximization, where the reward is a controllable function representing application performance. Our framework, unlike prior works, employs a DQL agent for joint learning of scheduling and congestion control for each DQL-MPTCP connection. The two controls (congestion and scheduling) jointly and effectively work together towards achieving the target reward. Furthermore, our implementation setup is versatile to use, i.e. whether the sender is a device or a server, readily inter-works with the current Linux kernel implementation (and makes comparisons with standard MPTCP algorithms fair).

 Unlike earlier model-driven MPCTP designs \cite{peng2016multipath, pokhrel2018improving, raiciu2011coupled, khalili2013mptcp}, our DQL-MPTCP utilizes runtime statistics of the underlying network environment. We combine the Deep Q-Learning (DQL) based control with policy gradients, as it is capable of effectively handling complex and dynamic state spaces. DQL-MPTCP learns dynamic system state and makes decisions without relying on any mathematical formulas or fixed control policies.  
DQL-MPTCP has to deal with mapping of traffic to multiple paths (subflows) efficiently. One may think of a simple approach and attach a DQL agent for each subflow independently, so as to separately conduct congestion control across all of the available paths. However, such approach is not fair to traditional single path TCP~\cite{raiciu2011coupled, peng2016multipath} and fails to capture the necessary and essential coupling of paths. DQL-MPTCP maximizes total performance reward by performing coupled congestion control and packet scheduling over all active paths/subflows of a connection. As will be discussed, we adopt a prioritized experience replaying concept to attain desired performance even with a limited data sets.

\subsection{Novelty and Contributions}

In this paper, we exploit a setting where a DQL agent access datasets from the history of the system~\cite{hester2018deep}; using prioritized experience replay~\cite{schaul2015prioritized} we store experience tuple (\emph{current state, current action, reward, next state}) in the memory, and then sample most important transitions from the memory more frequently.\footnote{For scalable design (i.e. memory requirements), we consider deleting the tuples after sampling and maintain only the samples.} As the users could be mobile and the number of available network paths (subflows joining and leaving) changes over time, existing DRL and convolution neural network approximation framework. To achieve an MPTCP algorithm that leverages few samples of data to accelerate the online learning even from relatively small datasets, we exploit a \emph{recurrent neural network} integrated with the \emph{long short-term memory} (LSTM) as a replay buffer, and continuously capture  packet flow dynamics. More specifically, our design is capable of automatically assessing and selecting the essential data sets from replay during learning. 

Our main contributions in this article are as follows:

\begin{itemize}
    \item [C$_1$] We propose a new approach for designing coupled multipath congestion control and packet scheduling using Deep Q-Learning with policy gradients.
    \end{itemize}
    In contrast to standard MPTCP designs, this new direction in designing multipath TCP performs both packet scheduling and congestion control. The two controls (congestion and scheduling) jointly and effectively work together towards maximizing the performance. Moreover, the aggregate reward can be selected as a function of packet goodput (rate of in-order received packets) and overall delay.
    
    \begin{itemize}
    \item [C$_{2}$] We develop an integrated `\emph{Policy Gradient with Deep Q-Learning}'  framework for multipath TCP, that uses an intelligent agent with the objective of maximizing a desired performance reward and performing dynamic congestion control and scheduling across all available network paths.  
    \end{itemize}
    Accommodation of accurate state and fine-grained actions results in a large state-action space; therefore, we integrate policy gradient algorithm with an actor-critic mechanism. To respond to highly dynamic networks, we exploit the agent`s policy distribution to average the value function over a set of actions, rather than only sampled action values. 
    
    To elaborate, when a connection is first initiated the end hosts {have limited updated knowledge of the characteristics of the path. \textcolor{black}{Typically, information is confined to that gleaned from the initial connection handshake and also from cached historical information. As the connection proceeds, feedback is obtained from packet transmissions but this feedback is delayed. The reason being 25-50 ms typical Round-Trip Times (RTTs) over local wireless internet paths, which can correspond to hundreds of packets `in flight' (around 150ms for international paths). Forward error correction has been used~\cite{8450632, pokhrellow} to ameliorate such circumstances. Short connections (i.e. small TCP flows), thus have limited information as to the path characteristics. Longer connections, e.g., bulk TCP flows, need to learn the path characteristics on the fly while transmitting packets. Our new design is to ameliorate the delay in learning path characteristics.} 
 \begin{itemize}    
     \item [C$_{3}$]  \textcolor{black}{In our design of DQL-MPTCP, we  utilize long short-term memory (LSTM) \textit{recurrent neural network} for tracking and learning of dynamic behaviour of wireless/wireline paths,  and appropriately adjusting congestion control and packet scheduling.}
     \end{itemize}
 \textcolor{black}{We investigate the design of intelligent schedulers that assign packets to paths in such a way as to minimise the reordering delay at the receiver. Key aspects that we address include a delay-aware scheduling, the impact of connection length of the optimal scheduler strategy and the need for joint scheduling of transmissions and learning of path characteristics (see Sec.~\ref{remark1} for details).}

\begin{itemize}
     \item [C$_{4}$] We conduct stability analysis of DQL-MPTCP, which provides important practical insights for the protocol design (Sec. IV). 
\end{itemize}

\begin{itemize}
     \item [C$_{5}$] We evaluate DQL-MPTCP and compare its performance with the standard MPTCP algorithms. Our comparison demonstrates significant  performance improvements with DQL-MPTCP. 
\end{itemize}


\section{Preliminaries and Motivations}

Multipath congestion control and packet scheduling mechanisms require continuous control action. However, it is impossible to apply Q-learning to continuous control action directly, because in continuous spaces, finding a greedy policy requires an optimization of control action $\bmath a_t$ at every time stamp. In fact, this type of optimization is sluggish in nature and hard to implement in practice with unconstrained approximates and large  (control action) state  spaces. Therefore, we use a deterministic policy gradient algorithm with an actor-critic approach. Such gradient-based algorithm maintains an actor (function) which can specify the policy by mapping states to a (specific) control action directly.  The critic can be learned by using the Bellman equation (similar to Q-learning). By applying the chain rule to the expected utility, the actor can be modified from an initial distribution with respect to the parameters. 

It is well-known that using a nonlinear approximation such as (deep) neural network for reinforcement learning is not desirable~\cite{li2019smartcc, xu2019experience}. In fact, such a nonlinear approximation is unstable and may lead to a diverging state~\cite{haarnoja2018learning, watkins1992q, achiam2019towards}. The pioneering work on the Deep Q-networks by Mnih \emph{et al.}~\cite{mnih2015human} laid the foundations for solving complex decision problems; their idea of combining experience replay with Q-learning and convolutional neural network enabled the framework to learn and perform intelligently. This approach not only avoids divergence and oscillations, but also is comparable to humans in making multilevel decisions. We exploit the ideas underlying the success of DQL~\cite{lillicrap2015continuous} in our continuous multipath congestion control and scheduling by combing it with policy gradients.

Recurrent neural networks exhibit a high capability of modeling nonlinear time series problems in an effective way; however, there are some issues to be addressed. In particular, they are  unable to train time varying network lags that are common in future wireless networks, and rely on a predetermined time duration for learning making them inflexible. To overcome the aforementioned shortcomings of the recurrent neural networks, we use LSTM, that act as a prioritized experience replay buffer for continuously tracking the active subflows and their interactions with the networks.
    
\subsection{Integrated Learning Framework}

We start by employing a Markov Decision Process (MDP), a 5-tuple  ($\mathbb S(t)$, $\mathbb A(t)$, $\bmath R(t)$, $\bmath T(t)$, $\gamma$), where $\mathbb S(t)$ is a finite set of states, $\mathbb A(t)$ is a finite set of actions, $\bmath R(t)$ is a reward function, $\bmath T(t)$ is the transition function ($\bmath T(t)$ determined from the probability state distribution $\bmath P(t)$) and $\gamma$ is a discount factor. In every state $s(t)\in\mathbb S(t)$, the agent takes an action $a(t)\in \mathbb A(t)$, receives a reward $\bmath R(a(t),s(t))$, and attains new state $s(t+1)$ obtained from the probability distribution $\bmath P(s(t+1)|s(t),a(t))$. The important notations used in our design and algorithms are summarized in Table~\ref{tab:symbols}.

\begin{table}
\caption{Summary of Notations. \label{tab:symbols}}
\begin{tabular}{l| p {6cm}}
\hline
\bf Notation & \bf Description \\ \hline
$\mathbb N(.)$, $\mathbb E(.) $, $\bmath Q(.)$& Representation $\mathbb N(.)$, Actor $\mathbb E(.) $ and Critic $\bmath Q(.)$\\
$\hat{\mathbb N}(.), \hat{\mathbb E}(.), \hat{\bmath Q}(.)$ & Target Representation, Actor and Critic Networks; Replicate the structure of $\mathbb N(.), \mathbb E(.)$ and $\bmath Q(.)$\\
$\mathbb S(t)$& Finite set of states, $s(t)\in\mathbb S(t)$, new state $s(t+1)$, $\mathbb S(t)=\{s(t,n)\}$ indicates subflow on link $n$\\
$\mathbb A(t)$& Finite set of actions, $a(t)\in \mathbb A(t)$\\ 
$\bmath R(t)$ & Reward function, a reward $\bmath R(a(t),s(t))$ from distribution $\bmath P(s(t+1)|s(t),a(t))$\\
$\bmath T(t)$ & Transition function determined from the probability state distribution $\bmath P(t)$\\
$\gamma$ & Discount factor.\\ 
$Q^{\star}(s,a)$ & Estimate of the expected future reward \\ 
$\Pi^{\star}$& Optimal policy  (greedy policy for $Q^{\star}$)\\
$\mathbb F(t)$& Final state for all subflows of the connection, ultimate final state $\bmath f^N(t)$\\
$\bar{\bmath\tau}=\{\bar\tau_n\}$&RTTs of all active subflows\\
$\mathbb W(t)=\{W_n\}$ & Size of the congestion windows\\ 
$U(t,n)$& Utility of the subflow $n$ at time $t$\\
$\hat{\bmath a}(t)$& Control action for
the target subflow\\ 
$y_i$& Target for the Critic $\bmath Q(.)$ \\\hline
\end{tabular}
\vspace{-5 mm}
\end {table}

We have a policy that specifies for every state which action the agent can take, where the aim of the agent is to discover the policy that potentially maximizes the expected reward. For tractability~\cite{winstein2013tcp}, we consider 6-tuple (\textit{sending rate, throughput, RTT, change in window, schedule, difference in RTT}) as the state of an MPTCP connection.\footnote{In our experiments we have observed that adding more parameters, increasing data sampling complexity with no noticeable improvement in the performance of the agent.} 

Our DQL-MPTCP takes action on each MPTCP subflow and provides what change (window increase, decrease and/or packet schedule change) needs to be made concurrently to the congestion windows and packet scheduling for the subflows.

 An estimate of the expected future reward that can be obtained from ($s(t),a(t)$) is given by the value (using Bellman equation)~\cite{hester2018deep}
\begin{equation}\tag{A1}
Q^{\star}(s,a)=\underset {s'\sim P} {\mathbb E}\big[R(s, a, s')+\gamma \max_{a'}Q^{\star}(s', a') \big],
\end{equation}
 The optimal policy, denoted by $\Pi^{\star}$ can be obtained as a greedy policy with respect to $Q^{\star}$:
\begin{equation}\tag{A2}
\Pi^{\star}(s)=\arg \max_{a'}Q^{\star}(s, a).
\end{equation}
It has been known that with respect to $Q^{\star}$ approximator, algorithm approximates greedy policy as the optimal policy. Our prioritized experience replay technique mandates the DQL agent to sample more frequently the most important state transitions from the memory. For example, we consider the likelihood of sampling a transition proportion to its priority in the spirit of those underlying \emph{learning from demonstrations}~\cite{lillicrap2015continuous, hester2018deep}.  Therefore, with the prioritized experience replay~\cite{schaul2015prioritized}, we develop a mechanism to attain comparatively better performance  even with a limited data sets and learning time.
 
\subsection{Joint Scheduling and Congestion Control}
\label{remark1}

A real momentum in recent years has been in the development of packet scheduling and congestion control mechanisms which aim to maximize the performance of MPTCP. See~\cite{peng2016multipath, chiariotti2019analysis, khalili2013mptcp, pokhrel2018improving, pokhrellow} for multipath congestion control~\cite{peng2016multipath, chiariotti2019analysis, khalili2013mptcp, pokhrel2018improving, pokhrellow, pokhrel2020multipath, 8666496} and~\cite{hurtig2018low, garcia2017low,  lim2017ecf} for the packet scheduling. Perhaps the most known scheduling policy is the one adopted by MPTCP scheduler~\cite{raiciu2011coupled}, the \textit{min--RTT},  which assigns packets to the smallest RTT path and fills its window, and so on to the other smallest RTT paths. 

\textcolor{black}{The time taken for a packet to traverse a mobile network path is always stochastic and time-varying (\textit{variable path delay}) as a result of queueing, other flows sharing the path, wireless link layer retransmissions etc.  When packets are sent via multiple paths they therefore can easily arrive at the destination reordered in which case they need to be buffered until they can be delivered in order to higher layers, leading to {head-of-line blocking}.  Moreover, the resulting buffering delay in MPTCP can be substantial~\cite{pokhrellow, 8450632, pokhrel2018improving}, to the point where it largely undermines the throughput gain from use of multiple paths.}\footnote{\textcolor{black}{In existing works~\cite{pokhrellow, 8450632, pokhrel2018improving}, the problem of \textit{variable path delay} was partially addressed by employing delay adaptation~\cite{pokhrel2018improving} and forward error correction~\cite{pokhrellow, 8450632} at the MPTCP source.}}

In this work, we propose a new intelligent \textit{Packet Scheduling} policy coupled tightly with multipath congestion control, with a notion that for controlled performance there is a potential to efficiently utilize the time-varying wireline and wireless links even when per-packet delays are highly fluctuating. To attain this, we have to move from the usual notion of considering packets individually to a new consideration of jointly scheduling collections of packets subject to an overall minimal delivery delay from all active subflows. 

This new notion of scheduling is driven by our observation that the QoE (quality of experience) requirement is usually to transmit application layer entities such as video clips, icons, web pages, micro-blogs, etc. with the lowest attainable delay, and \textit{it is the overall aggregate delay rather than the per-packet delay which is critical} in such context. This observation has fundamental implications in our packet scheduler design. We are scheduling packets in groups; therefore, such scheduling is capable of considering application layer contexts and their relationship to path uncertainties.

The packet scheduling policy in this work is coupled tightly with the multipath congestion control mechanism. Our scheduler assigns a number of packages to each path proportional to their short term average goodput (\textit{number of successful packets per RTT}).\footnote{The transient behaviour, i.e.,  the trajectories through which short term average goodputs ($\theta_n(t)$) approach their steady values, is essential for the computation of exact reordering delay. However, we are concerned with the difference in short term steady goodputs, and our interest lies in designing a solution that can alleviate the severe reordering delay perceived at the application layer. } Given the availability in the windows (specified by the multipath congestion control process), our scheduler always assigns packets to the paths proportionally, ranging from the largest goodput path to the small goodput path. Such a scheduling policy has the  effect of minimizing packet re-sequencing at the receiver~\cite{saha2019musher}.

\section{The Proposed MPTCP Scheme}
\label{sec:design}
In this section, we present the design of our \emph{Deep Q-Learning enabled MPTCP with policy gradients} for joint congestion control and packet scheduling across paths of different characteristics. At a high level, the interactions between modules of the proposed DQL-MPTCP are as illustrated in Figure~\ref{fig:MPTCP-blocks}. A  detailed  explanation  of  actor critic training, representation network and subflow state {analysis} is provided with algorithms in Section 3.1-3.3. As discussed in Section 3.3, the DQL agent  (submodules inside the dashed portion in Figure~\ref{fig:MPTCP-blocks})
of the proposed MPTCP interacts with the links and the user to collect the freshest information of the state $\mathbb S(t)=\{s(t,n)\}$ here, $n$ indicates an MPTCP subflow on the link $n$. 
At the beginning of time  slot $t$, the agent computes the reward, $\bmath R(t)$ (the sum of utilities of all active subflows from previous actions), by using the actor-critic network (Sec.~\ref{sec:accrt}). The computation is based on the representation learned (Sec.~\ref{sec:netrep}) by the LSTM and the current state of the subflows $\mathbb S(t)$ (final states $\mathbb F(t)$ and $\bmath a(t)$  are used to compute $\bmath R(t)$).  

The DQL-agent is queried periodically by the MPTCP source (one query per slot), in order to i) update the size of the windows ($\mathbb W(t+1)$), ii) update the RTTs ($\bar{\bmath\tau}=\{\bar\tau_n\}$) of active subflows and iii) schedule packets to the subflows (\textit{schedule}, for next slot, see Figure~\ref{fig:MPTCP-blocks}). Rather than using only the actions that were actually executed, we consider a policy gradient algorithm that uses the agent's explicit representation of all action values to estimate the gradient of the policy. Therefore, our framework implements
the action (via packet scheduling and congestion control at
the MPTCP source in the kernel) by observing the reward $R(t)$. To ensure stability and maintain replay buffer, 
we have used target Network Representation $\mathbb{\hat{N}}(.)$ Actor $\mathbb{\hat{E}}(.)$ and Critic  $\bmath{\hat{Q}}(.)$, which replicate the structure of their corresponding networks in the state analysis, i.e., Network Representation $\mathbb N(.)$, Actor $\mathbb E(.) $ and Critic $\bmath Q(.)$ respectively.

\noindent
\begin{figure}[t]
\centering
\includegraphics[width= 2.65 in]{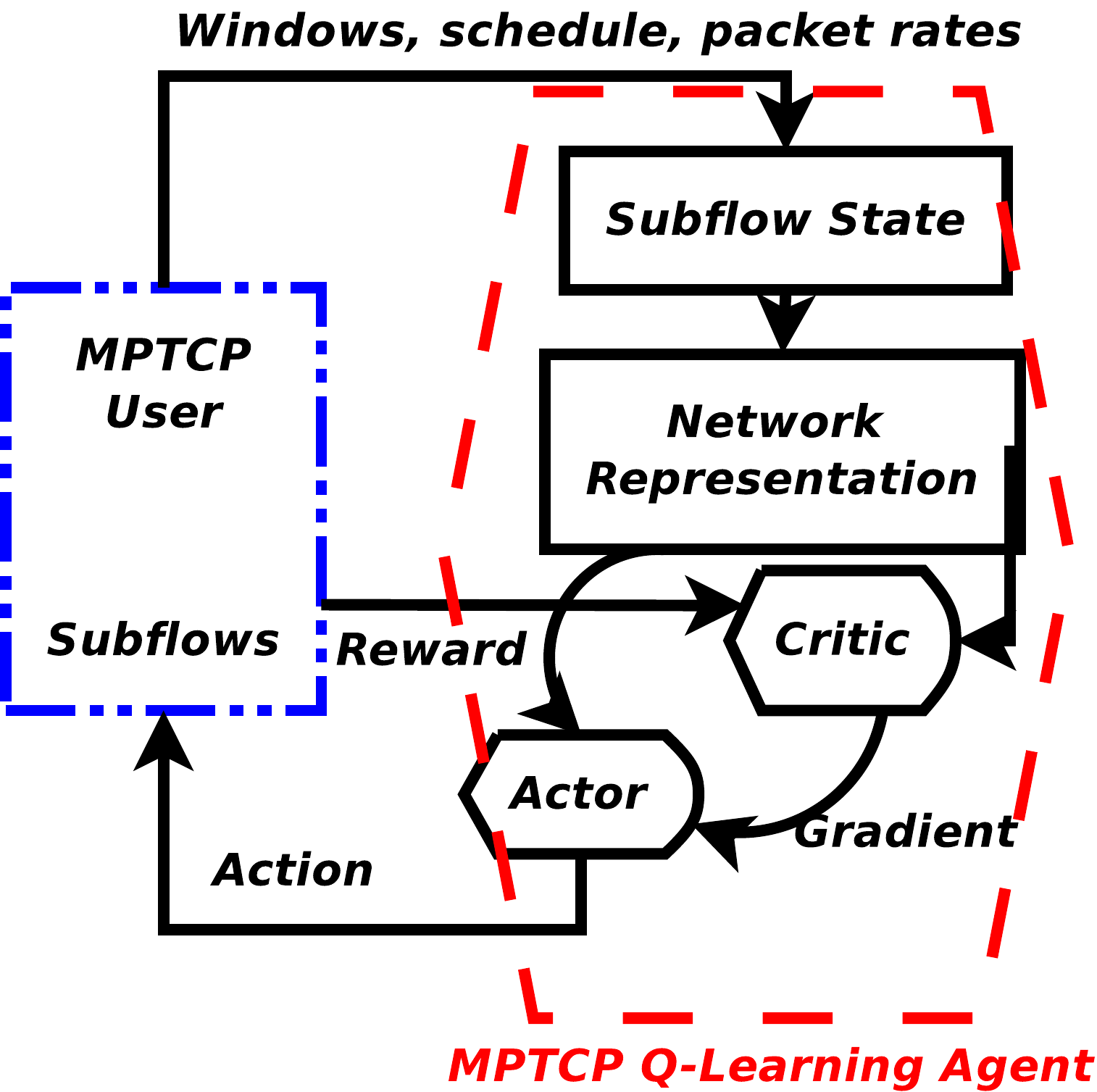}
\caption{\rm A abstract view of interrelationship between the DQL-MPTCP modules. The modules inside the dashed region (form a DQL Agent) take  packet rates, reward, congestion windows as input and generate control actions as outputs. A module on the left rectangle (is a MPTCP user with source program at the OS kernel), which provides new packet rates, schedule and congestion windows. \label{fig:MPTCP-blocks}}
\end{figure}

A key feature of DQL is its use of a target network whose purpose is to stabilize the learning process. In traditional Q-learning, the value of executing an action in the current state is updated using the values of executing actions in the next state. This process can be unstable since the values on both sides of the update equation can change at the same time. The target network keeps a copy of the estimated value function to serve as a stable target for a few future steps~\cite{mnih2015human}. The main building blocks in Figure~\ref{fig:MPTCP-blocks} and  our algorithms are explained next.

\subsection{Network Representation}
\label{sec:netrep}
The Network Representation module creates a (representation) vector simply by observing the states of MPTCP subflows ($\bmath s(t,n)$ as inputs). However, it is challenging to deal with the representation when the number of subflows are time-varying (joining and leaving the MPTCP connection), which will be common in future wireless networks. Rather than using fixed input size Deep Neural Network approach (e.g., feed forward), we use LSTM which has the capability to handle variable number of inputs over time so as to capture the dynamics of varying number of MPTCP subflows (corresponding to the conditions of the available network paths and access interfaces). The approach is to provide the states of the subflows one by one into the LSTM for learning the Network (and Delay) Representation in a sequential manner.

The ultimate final state $\bmath f^N(t)$,  the output of this module (we denote $\mathbb F(t)$ for all subflows of the connection), is then fed for Actor-Critic training. 
Furthermore, we propose to train both `LSTM-based representation' and `Actor-Critic Network' jointly by using backward propagation. We have an inner loop between three modules, viz. \emph{subflow state analysis $\rightarrow$ network representation $\rightarrow$ actor-critic training}. It is worth noting that jointly training the interrelated modules (by observing the state of the subflows) yields  improved performance rather than training them individually. 

\subsection{Actor-Critic Training}
\label{sec:accrt}
The actor network  consists of completely connected
LSTM with two hidden layers,
that consists of 128 neurons in both layers. The critic is very similar to the actor except of the output layer with a 
single linear neuron. We use Rectified
Linear function in hidden layers and
hyperbolic tangent function in the
output layer for activation. 

The final state $\mathbb F(t)$ produced by the LSTM-based representation is combined with the state of the corresponding MPTCP subflow, which is then supplied to the Actor-Critic network. 
The agent specifies at each time $t$ how to update the size of the congestion window ($\mathbb W(t)$) and changes schedule (if any) for each subflows of the MPTCP connection in order to maximize the reward. 

The reward in our framework is the sum of the utilities of all active subflows belonging to the MPTCP connection, i.e., $ R(t):=\sum_n^N U(t,n)$ where, $U(t,n)$ is the utility of the subflow $n$ at time $t$. It is worth noting that our framework is flexible to define and handle desired utility function (e.g., fairness-based~\cite{pokhrel2019fair} or, loss, delay, and throughput-based~\cite{dong2018pcc}) and scheduling policy.

\textcolor{black}{Resources can be allocated amongst the competing flows according to various policies. To perform MPTCP congestion control and scheduling, we aim to facilitate applications with efficiency and fairness. Although there is no global notion of fairness, proportional and max-min fair allocations~\cite{kelly1998rate, mazumdar1991fairness} are the two well-known schemes.}

\textcolor{black}{\noindent \textit {Proportional Fairness~\cite[Eqn.(2.10), pp. 37]{pokhrel2017modeling}: For a set of users $\mathcal U = \{1, 2, . . . , n\}$ the goodput allocation [$\theta_1^*, \theta_2^*, . . . , \theta_n^*$] to the set of users $\mathcal U$ is proportionally fair if and only if for any other feasible scheme
[$\theta_1, \theta_2, . . . , \theta_n$], such that}}

\begin{equation}
   \textcolor{black}{\sum_{u\in\mathcal U} \frac{\theta_u-\theta_u^*}{\theta_u^*}\leq 0.}
   \label{eqn:prpp}
\end{equation}

\textcolor{black}{In~\eqref{eqn:prpp}, when the proportional change in one user's packet rate is positive, at least another user for which the change in packet rate is negative.}

In our implementation, we have used the well-known proportional fairness approach for all subflows (active paths) and maximize the function $U(t,n)=log \theta_n(t)$, where $\theta_n(t)$ is the short-term average goodput along path $n$ in the previous slot perceived by the subflow. For delay minimization, our scheduling policy (recall Section~\ref{remark1}) implements batch scheduling of packets to the subflows with a rate proportional to their short-term average goodputs.

\begin{remark}
Despite its empirical success and advancement, the theoretical perception behind the convergence of the actor-critic algorithm is lagging. Actor-Critic learning can be viewed as an alternating online bilevel optimisation process, the convergence of which is known to be fragile. In Sec.~4, we concentrate on applying DQL agent regulators, a basic but essential determination in Q-learning with gradients, to understand the learning dynamics' instability. Our aim in this work is to design a flexible framework, the stability of which is independent of the underlying utility function. Therefore, in this setting, we analyse a non-asymptotic convergence of the DQL-learning and demonstrate how it seeks a linear rate of convergence that is globally optimal. Our stability analyses in Sec.~4 can be the first step to thoroughly comprehend bilevel optimisation for the utility-based actor-critic problems, which will be investigated in future (in the worst case this can be NP-hard, non-convex and is mostly overcome with heuristics). 
\end{remark}
\begin{algorithm}[!t]
\caption{Training Algorithm}
\label{Training algorithm}
\begin{algorithmic}[1]
\Procedure{ Training} {}\\
Sample transitions ($\mathbb S_i, \bmath{{a}}_i, \bmath R_i, \mathbb S_{i+1} $) from priority replay\\
Compute $\bmath f^N_{i+1}$ using $\mathbb{\hat{N}}(\mathbb S_{i+1})$\\
Evaluate target for $\bmath Q(.)$, ${\footnotesize y_i= \bmath R_i+\gamma \bmath{\hat{Q}}(\bmath f^N_{i+1}, \mathbb{\hat{E}}(\bmath f^N_{i+1}) )}$\\
Critic parameters $\min(loss)$, ${\footnotesize 1/k\sum_i^k(y_i-\bmath Q(\bmath{{a}}_i, \bmath {f^N_{i}} ))^2}$\\
Compute Policy Gradient from Critic: $\bigtriangledown_a\bmath Q(\bmath{{a}}_i, \bmath {f^N_{i} })$\\
Update actor parameters, ${\footnotesize 1/k\sum_i^k\bigtriangledown_a\bmath Q(\bmath{{a}}_i, \bmath {f^N_{i}}). \bigtriangledown_{k}\mathbb E(\bmath{f^N_{i}})}$\\
Compute Policy gradient from Actor: $\bigtriangledown_k\mathbb E(\bmath{f^N_{i}} )$\\
Update network representation (policy gradient),\ \ \ \ \  \ \;\;\;\; ${\footnotesize {1/k\sum_i^k\bigtriangledown_{\bmath{a}}\bmath Q(\bmath{{a}}_i, \bmath {f^N_{i}}). \bigtriangledown_k\mathbb E(\bmath {f^N_{i}}).\bigtriangledown_{\bmath R_i}\bmath R(\mathbb{S}_{i+1}}) }$
 {\EndProcedure}
\end{algorithmic}
\end{algorithm}

\subsection{DQL-MPTCP with Policy Gradients}
\label{sec:algo12}
Our DQL-based multipath congestion control and scheduling mechanism is illustrated in Algorithms~\ref{Training algorithm} and~\ref{MPTCPalgorithm}. Algorithm~\ref{Training algorithm} illustrates the training process. Algorithm~\ref{MPTCPalgorithm} initializes the parameters of Representation Network  $\mathbb N(.)$, Critic networks $\bmath Q(.)$ and  Actor networks $\mathbb E(.) $. In order to ensure stable learning (see step 3, Algorithm~\ref{MPTCPalgorithm}) , we have used target networks $\mathbb{\hat{N}}(.)$ $\mathbb{\hat{E}}(.) $ $\bmath{\hat{Q}}(.)$, which replicate the structure of their corresponding networks, i.e., $\mathbb N(.)$, Actor $\mathbb E(.) $ and Critic $\bmath Q(.)$ respectively.

 For the convergence of the learning process, the target network parameters are updated using a smallish control parameter (0.001, see steps 17-19, Algorithm~\ref{MPTCPalgorithm}), such that the target parameters are adapted slowly in each iteration--recall that the target network, by design, needs to update slowly for stability.\footnote{The rationale for selecting smallish control parameter values comes from our stability analysis in Section~\ref{sec:stab}.} Our MPTCP DQL agent runs all the time, listening for periodic queries from the MPTCP source (implemented in the kernel). 
 
 To enable exploration, which is useful when training inexperienced DQL agent, 
we add correlated noise, using the Ornstein-Uhlenbeck stochastic process~\cite{lillicrap2015continuous, hester2018deep} to control actions in run-time. Such coupling of prioritized experience replay with exploration helps the DQL agent improve its actions of adjusting subflow window sizes and packet scheduling.\footnote{An effective way to conduct retraining of the agent to adapt for the new network settings will benefits from transfer learning~\cite{pokhrel2021multipath} and needs further investigations in future.}

Recall that the actor and critic networks  are completely connected
LSTM with two hidden layers,
that consists of 128 neurons in both layers. We use Rectified
Linear function in the hidden layers and
hyperbolic tangent function in the
output layer for activation. The  critic and actor networks
are trained jointly using Adam optimizer (\emph{we set learning rates to 0.001 and 0.01 respectively and discount factor $\gamma= 0.95$}).

The final state network representation of all active subflows $\mathbb{F}(.)$ is derived from the representation network $\mathbb{N}(.)$ (see step 8, Algorithm.~\ref{MPTCPalgorithm}), and the control action for
the target MPTCP subflow $\hat{\bmath a}(t)$ is computed by using the actor $\mathbb{E}(.)$ (see step 9,  Algorithm.~\ref{MPTCPalgorithm}). First, the generated transition samples are stored into kernel (memory), thereafter they are randomly sampled for training the tuple
(\emph{network representation $\mathbb N$, actor $\mathbb E$, critic $\bmath Q$}) jointly by using $k$  prioritized samples (step 2, Algorithm.~\ref{Training algorithm}). The critic is a \emph{Deep Q-Learning Network} and its parameters are updated by minimizing the squared error (step 5, Algorithm.~\ref{Training algorithm}), i.e., the target for
critic $y_i$ is evaluated by applying the Bellman equation  (step
4, Algorithm.~\ref{Training algorithm}). The Q-function uses the Bellman equation and takes  action ($\bmath{a}(t)$) and state ($\mathbb S(t)$)  as inputs. For continuous control, the parameters of the network representation and actor networks are adapted together with the policy gradients using the chain rule~\cite[Eqn. (6)]{lillicrap2015continuous} by using $k$ samples (steps 6-10, Algorithm.~\ref{Training algorithm}). Our DQL framework consisting of the representation, critic and actor networks has key ingredients for training and learning MPTCP packet scheduling and congestion control. 

\begin{remark}
The complexity analysis of DQL algorithm is another evolving research direction. To the best of our knowledge, despite the popularity of DQL with policy gradients, we are still investigating to accurately predict the computational complexity to train/learn a DQL network and solve a given problem.  With the relevant insights from this work and~\cite{kumar2019sample, degris2012off}, an important future research direction would be a thorough complexity analysis of DQL-MPTCP.
\end{remark}
\begin{algorithm}[t]
\caption{DQL-MPTCP Coupled Algorithm}
\label{MPTCPalgorithm}
\begin{algorithmic}[1]
\Procedure{ Begin} {}\\
Representation Network $\mathbb N(.)$, Actor $\mathbb E(.) $ and Critic $\bmath Q(.)$\\
Target Networks $\mathbb{\hat{N}}(.)$ $\mathbb{\hat{E}}(.) $ $\bmath{\hat{Q}}(.)$\\
Ornstein-Uhlenbeck process $\mathbb a$ for exploration
\EndProcedure
\Procedure{Scheduling $\&$ Congestion Control} {}\\
\textbf{While} {(Schedule Change || Window Adapt)} \textbf{do}\\
\hspace{5 mm}Compute final state $\mathbb F(.)$ and schedule by $\mathbb N(.)$\\
\hspace{5 mm}Update target $\bmath{\hat{a}}(t)$ using Actor $\mathbb E(\bmath f^N(t))$\\
\hspace{5 mm}Create action $\bmath{{a}}(t)$ applying  $\bmath{\hat{a}}(t)$ and process $\mathbb a$\\
\hspace{5 mm}Take action $\bmath{{a}}(t)$, observe $\bmath R(.)$ and $\mathbb S(t+1)$\\
\hspace{5 mm}Store transitions ($\bmath{{a}}(t), \mathbb S(t), \bmath R(t), \mathbb S(t+1) $)

\Procedure{ on-line Training} {}\\
 \hspace{8 mm}Call { \sc Training}
\EndProcedure
\Procedure{Update Target Network} {}\\
\hspace{8 mm}$\mathbb{\hat{N}}(.):=.001\mathbb{{N}}(.)+.999\mathbb{\hat{N}}(.)$\\ \hspace{8 mm}$\mathbb{\hat{E}}(.):=.001\mathbb{{E}}(.)+.999\mathbb{\hat{E}}(.)$\\ \hspace{8 mm}$\bmath{\hat{Q}}(.):=.001\mathbb{{C}}(.)+.999\bmath{\hat{Q}}(.)$
\EndProcedure\\
\textbf{endWhile}
\EndProcedure
\end{algorithmic}
\end{algorithm}

\section{Stability Analysis} 
\label{sec:stab}
Peng \emph{et al.}~\cite{peng2016multipath} conducted a detailed study of the various variants of the MPTCP using a generalised fluid model and analyzed their stabilization aspects. In fact, analyzing stability has been an essential first step in designing a congestion control algorithm to guarantee it has a desirable equilibrium and convergence properties. Our DQL-MPTCP algorithm is first trained offline (using Algorithm.~\ref{Training algorithm}) and then deployed, and goes along to learn online (step 14 in Algorithm.~\ref{MPTCPalgorithm} calls Algorithm.~\ref{Training algorithm}). Therefore, we examine whether the learning process of DQL-MPTCP will indeed drive the network towards an equilibrium starting from an arbitrary initial state. Even though in reality a network is seldom in equilibrium, a stable online learning process ensures that it is always pursuing a desirable state, which also makes it easier to understand the global protocol behavior of the overall network. 

{ For improving smoothness and responsiveness in the implementation, 
we indirectly guarantee the stability by imposing few constraints in the underlying Q-learning process of our MPTCP algorithm.} Given the convergence of value iteration (the Q-learning process), the policy iteration is guaranteed to converge. Further, at the convergence point, the current policy and its value function are the optimal policy and the optimal value function. Recall (A2), the optimal policy $\Pi^{\star}$ is a greedy policy corresponding to $Q^{\star}$. In fact, any unconstrained policy iteration process leverages wider state space and region.  Besides, the convergence of the policy iteration has been proven to be not slower than that of the value iteration, therefore, policy iteration has the potential to improve both smoothness and responsiveness.
Based on above discussion, our focus in this section is mainly to quantify the convergence conditions of value iteration for the Q-learning process of our DQL-MPTCP.} 

Divergence is an important challenge in Deep Q Learning; it is not well understood and often arises in implementations. Therefore, it is important to analyze the stability of our proposed DQL-MPTCP algorithm~\cite{haarnoja2018learning, achiam2019towards}.

With relevant insights from~\cite{achiam2019towards}, the DQL ideas implemented in Algorithm~1 is to learn an approximate to the optimal value function $Q^{\star}$, which satisfies (A1). With $\mathcal T^{\star} Q(s,a)$ given by the right hand side of (A1), after $\mathcal T^{\star}: \mathcal Q\rightarrow \mathcal Q$ operation on Q functions; then (A1) is given by
\begin{equation}\tag{A3}
    Q^{\star}=\mathcal T^{\star} Q^{\star},
\end{equation}
where $\mathcal T^{\star}$ is the optimal Bellman operator, which  with modulus $\gamma$ is a contraction in the supremum norm. As a result, we have two cases.

{\sc Case I}. Given that the Q-function is represented by a finite table with (completely) known $\bmath R_i$ (reward function) and transition kernel (steps 12 and 14 Algorithm~1), $\mathcal T^{\star}$ can be determined and therefore, $Q^{\star}$ can be estimated by using the successful method for computing an optimal Markov Decision Process policy (and its value), known as \emph{value iteration}. For such cases, the value iteration starts at the end and then works backward, refining the estimate of $Q^{\star}$. 
Consider  $Q^{i}$ be the Q-function assuming there are $i$ stages to go, then, these can be defined recursively; and the iteration starts with an arbitrary function $Q_{0}$ and uses the following equation to obtain the function for $i+1$ stages (to move from the function for $i$ stages to go): 
\begin{equation}\tag{A4}
    Q_{i+1}=\mathcal T^{\star} Q_{i}.
\end{equation}
The convergence of the value iteration given by (A4) using $Q_{0}$ as an initial point is undertaken by the \emph{Banach Fixed Point Theorem}.

{\sc Case II}. Given that $\bmath R_i$ and the transition kernel are partially known (not completely known), it is viable to use Q-learning and learn $Q^{\star}$ in such settings~\cite{watkins1992q, achiam2019towards}. In fact, Watkins \emph{et al.}~\cite{watkins1992q} demonstrated that  Q-learning  converges  to  the  optimum  $\bmath a^{\star}$ (action-values)  with  probability $1$  so  long  as  all  actions  are  repeatedly  sampled  in  all  stages  and  the  action-values  are  represented  discretely (steps 10, Algorithm.~\ref{MPTCPalgorithm} and 4-10, Algorithm.~\ref{Training algorithm}). The Q-values of the ($s,a$) pairs are updated using reward and next state for estimating $\mathcal T^{\star} Q_{i}(s,a)$ given by \[
\mathcal T^{\star} Q_{i}(s,a)= \bmath R_i +\gamma \max_{\hat a} Q_{i}(\hat s,\hat a)\] as:
\begin{equation}\tag{A5}
Q_{i+1}(s,a)=Q_{i}(s,a)+\alpha_i \big(\mathcal T^{\star} Q_{i}(s,a)-Q_{i}(s,a)\big).
\end{equation}
(A5) under much relaxed condition, viz., \emph{ the learning rates $\alpha_i\in[0,1\}$ must approaches zero and all state action pairs must be visited often enough}, converges to $Q^{\star}$. One can observe that the Q-learning given by (A5) works on time-varying difference as the updates in each stages are based on the temporal difference ($\mathcal T^{\star} Q(s_t,a_t)-Q(s_t,a_t)$) given by:
\begin{equation}\tag{A6}
\Delta_t=\bmath R_t +\gamma \max_{\hat a} Q(\hat s_t,\hat a_t)-Q(s_t,a_t).\;\;
\end{equation}

Motivated by the analysis in~\cite{achiam2019towards}, our DQL-MPTCP is based on the generalization of (A5) to the approximation setting:
\begin{equation}\tag{A7}
g'=g+\alpha \big(\mathcal T^{\star} Q_{g}(s,a)-Q_{g}(s,a)\big)\nabla_gQ_{g}(s,a)\big],
\end{equation}
where $Q_{g}$ is a continuous function with parameters $g$. Observe that (A7) becomes (A5) given $Q_{g}$ is a table. On the whole, such DQL representation uses gradient descent and experience replay, maintaining the expected update as:
\begin{equation}\tag{A8}
g'=g+\alpha \underset {s,a\sim \rho} {\mathbb E}\big[\big(\mathcal T^{\star} Q_{g}(s,a)-Q_{g}(s,a)\big)\nabla_gQ_{g}(s,a)\big],
\end{equation}
where $\rho$, at the time of updates, is the distribution of the prioritized experience in the replay.
Typically, to ensure stability, it is acceptable to replace, $\mathcal T^{\star} Q_{g}$ with the one based on slowly-updating target network, $\mathcal T^{\star} Q_{h}$, where $h$ is obtained by Polyak averaging~\cite{lillicrap2015continuous,haarnoja2018learning, watkins1992q, achiam2019towards}.\footnote{For simplicity and tractability~\cite{lillicrap2015continuous}, we omit the target network in our stability analysis.}

Using Taylor expansion of $Q$ around $g$ for the pair ($\hat s,\hat a$) we can observe the new Q-values based on (A8) as:
\begin{equation}\tag{A9}
Q_g'(\hat s,\hat a)=Q_g(\hat s,\hat a)+\nabla_gQ_{g}(\hat s,\hat a)^T(g'-g)+\BigO{\norm {g'-g}^2},
\end{equation}
Over a finite state action space, we consider matrix vector form in $\mathbb R^{|\mathbb S||\mathbb A|}$, therefore, combining (A8) and (A9) provides:
\begin{equation}\tag{A10}
Q_g'=Q_g+\alpha \mathcal K_gD_\rho\big(\mathcal T^{\star} Q_{g}(s,a)-Q_{g}(s,a)\big)+\BigO{\norm {g'-g}^2},
\end{equation}
where $D_\rho$ is a diagonal matrix obtained from the distribution from replay, $\rho(s,a)$ and $\mathcal K_g$ is $|\mathbb S||\mathbb A|\times|\mathbb S||\mathbb A|$ a matrix given by\[\mathcal K_g (\hat s, \hat a, s,a)=\nabla_gQ_{g}(\hat s,\hat a)^T\nabla_gQ_{g}(s,a).\]

Finally, our formulation (A10) provides us important insights to analyse the stability of DQL-MPTCP. In particular, the only  condition essential now for convergence of DQL-MPTCP is to guarantee that the update operator $\mathcal U: \mathcal Q\rightarrow \mathcal Q$ with
\begin{equation}\tag{A11}
\mathcal U Q_g=Q_g+\alpha \mathcal K_gD_\rho\big(\mathcal T^{\star} Q_{g}(s,a)-Q_{g}(s,a)\big)
\end{equation}
is a contraction on $\mathcal Q$.
With relevant insights from~\cite{haarnoja2018learning, watkins1992q, achiam2019towards}, this will be considered next. 

\subsection{Practical Design Insights}
\label{sec:stab1}
In this subsection, we investigate the potential conditions and study how the update $\mathcal U: \mathcal Q\rightarrow \mathcal Q$ given by (A11) may give rise to instability in DQL-MPTCP and how to repair such instabilities. Next,  motivated by the analysis in~\cite{achiam2019towards}, we decompose our analysis into following cases.

{\sc Case I}. We consider $\mathcal U^i$, a special case of $\mathcal U$ where $\mathcal K_g=1$ and $D_\rho=1$, therefore, (A11) becomes
\begin{equation}\tag{A11.1}
\mathcal U^i Q_g=Q_g+\alpha \big(\mathcal T^{\star} Q_{g}-Q_{g}\big).
\end{equation}
\begin{theorem}
$\mathcal U^i$ defined by (A11.1) is a contraction on $\mathcal Q$ and $Q^{\star}$ is the fixed-point.
\label{thm:1}
\end{theorem}
\begin{proof}
Observe that $\mathcal U^i$ given by (A11.1) satisfies
\begin{eqnarray}
&&\norm{\mathcal U^i Q_1-\mathcal U^i Q_2}_{\infty}\nonumber\\&=& \norm{\alpha(\mathcal T^{\star}Q_1-\mathcal T^{\star}Q_2)+(1-\alpha)(Q_1-Q_2)}_{\infty},\nonumber\\
&\leq&\alpha\norm{\mathcal T^{\star}Q_1-\mathcal T^{\star}Q_2}_{\infty}+(1-\alpha)\norm{(Q_1-Q_2)}_{\infty},\nonumber\;\;\;\;\;\\
&\leq&\alpha\gamma\norm{Q_1-Q_2}_{\infty}+(1-\alpha)\norm{(Q_1-Q_2)}_{\infty},\nonumber\\
&=&\norm{Q_1-Q_2}_{\infty}(1-\alpha(1-\gamma)),\nonumber
\end{eqnarray}
and since $1>(1 - \alpha(1-\gamma))$, the update $\mathcal U^i$ contracts and $Q^{\star}$ is its fixed-point (follows straightforward using (A3)).
\end{proof}
{\sc Case 1} and Theorem~\ref{thm:1} provide us the following important design insight for stable MPTCP.

{\sc Design Insight 1}. Given $\mathcal U$ of the proposed DQL-MPTCP gets increasingly closer to $\mathcal U^i$, one can anticipate progressively more stable DQL-MPTCP behavior.

{\sc Case II}. We consider $\mathcal U^{ii}$, a special case of $\mathcal U$ where $\mathcal K_g=1$, therefore, (A11) becomes
\begin{equation}\tag{A11.2}
\mathcal U^{ii} Q_g=Q_g+\alpha D_\rho \big(\mathcal T^{\star} Q_{g}-Q_{g}\big).
\end{equation}

\begin{theorem}
$\mathcal U^{ii}$ defined by (A11.2) is a contraction on $\mathcal Q$ and $Q^{\star}$ is the fixed-point, if $\rho(s,a)>0, \forall (s,a)$ and $\alpha\in(0,1/\hat\rho)$ where $\hat\rho=\max_{s,a}\rho(s,a)$.
\label{thm:2}
\end{theorem}
\begin{proof}
One can observe that for any ($s,a$),
\begin{eqnarray}
&&[\mathcal U^{ii} Q_1-\mathcal U^{ii} Q_2](s,a)\nonumber\\&=& \alpha\rho(s,a)\big([\mathcal T^{\star}Q_1-\mathcal T^{\star}Q_2](s,a)\big)\nonumber\\&&+(1-\alpha\rho(s,a))\big(Q_1(s,a)-Q_2(s,a)\big),\nonumber\\
&\leq&\alpha\gamma\rho(s,a)\norm{Q_1-Q_2}_{\infty}+(1-\alpha\rho(s,a))\norm{(Q_1-Q_2)}_{\infty}\nonumber\;\;\;\;\;\\
&=&\norm{Q_1-Q_2}_{\infty}(1-\alpha\rho(s,a)(1-\gamma)).\nonumber
\end{eqnarray}
By taking the $\max_{s,a}$ on both sides,
\begin{eqnarray}
&&\norm{\mathcal U^{ii} Q_1-\mathcal U^{ii} Q_2}_{\infty}\nonumber\\
&\leq&\max_{s,a}\norm{Q_1-Q_2}_{\infty}(1-\alpha\rho(s,a)(1-\gamma)).\nonumber\\
&=&\norm{Q_1-Q_2}_{\infty}(1-\alpha\tilde\rho(1-\gamma)),\nonumber
\end{eqnarray}
where $\tilde\rho=\min_{s,a}\rho(s,a)$. Note that the condition, $\rho(s,a)>0, \forall (s,a)$, is equivalent to $\tilde\rho>0$. Therefore, $1>(1 - \alpha\tilde\rho(1-\gamma))$, which mandates that the update $\mathcal U^{ii}$ contracts and $Q^{\star}$ is its fixed-point (follows straightforward using (A3)). However, observe that when $\tilde\rho=0$, we simply have an upper bound on $\norm{\mathcal U^{ii} Q_1-\mathcal U^{ii} Q_2}_{\infty}$. See~\cite{achiam2019towards} for more background.
\end{proof}

Considering $\mathcal U^{ii}$, we observe that the missing data in the distribution has adverse impact on the convergence of the learning process. Given that the (exploration) policy explores all state action pairs enough, $\mathcal U^{ii}$ behaves as expected, however, missing data may cause a problem. This observation provides us another important design insight.

{\sc Design Insight 2}. DQL-MPTCP may struggle to converge, when data is scarce at the beginning of the training, where initial conditions matter a lot.

The importance of\;{\sc Design Insight 2}  will be discussed in Section~\ref{sec:inst}.

{\sc Case 3}. We consider $\mathcal U^{iii}$, a special case of $\mathcal U$ where $\mathcal K_g=\mathcal K$ is positive-definite constant symmetric matrix, therefore, (A11) becomes
\begin{equation}\tag{A11.3}
\mathcal U^{iii} Q_g=Q_g+\alpha \mathcal K D_\rho \big(\mathcal T^{\star} Q_{g}-Q_{g}\big).
\end{equation}
Therefore, $\mathcal U^{iii}$ is the case of linear function approximation, $\mathcal K$ is constant with respect to $g$ and the Q-values before and after updates are given by
\begin{equation}
\mathcal Q_{g'}=\mathcal U^{iii} Q_g; \;\mbox{where}\; \nonumber
\end{equation}
\[\mathcal K (\hat s, \hat a, s,a)=\phi(\hat s,\hat a)^T\phi(s,a).\]
\begin{theorem}
$\mathcal U^{iii}$ defined by (A11.3) is a contraction on $\mathcal Q$ and $Q^{\star}$ is the fixed-point, if and only if i) $\alpha\rho_x\mathcal K_{xx}>1$, $\forall x$ and \[ii) \;(1-\gamma)\rho_x\mathcal K_{xx}\geq(1+\gamma)\sum_{x\neq y}\rho_y\abs{\mathcal K_{xy}}, \forall x,\]
where $x,y$ are the indices of state action pairs.
\label{thm:3}
\end{theorem}
\begin{proof} By using index notation
\begin{eqnarray}
&&[\mathcal U^{iii} Q_1-\mathcal U^{iii} Q_2]_x\nonumber\\&=& \alpha\sum_{y}\rho_y\mathcal K_{xy}\big([(\mathcal T^{\star}Q_1-Q_1)-(\mathcal T^{\star}Q_2\nonumber\\&&-Q_2)]_y\big)+[Q_1-Q_2]_x,\nonumber\\
&=&\alpha\sum_{y}\rho_y\mathcal K_{xy}[\mathcal T^{\star}Q_1-\mathcal T^{\star}Q_2]_y\nonumber\\&&+\sum_{y}(\Delta{xy}-\alpha\rho_y\mathcal K_{xy})[Q_1-Q_2]_y.\nonumber\\
&\leq&\sum_{y}\big(\alpha\gamma\rho_y\abs{\mathcal K_{xy}}+\abs{\Delta{xy}-\alpha\rho_y\mathcal K_{xy}}\big)\norm{Q_1-Q_2}_\infty.\nonumber
\end{eqnarray}
Let
\begin{eqnarray}
G(\mathcal K)=\max_x\sum_{y}\big(\alpha\gamma\rho_y\abs{\mathcal K_{xy}}+\abs{\Delta{xy}-\alpha\rho_y\mathcal K_{xy}}\big).\nonumber
\end{eqnarray}
Then, assuming $\alpha\rho_x\mathcal K_{xx}>1$, condition i),
\begin{eqnarray}
G(\mathcal K)&=&\max_x\Big(\alpha(1+\gamma)\sum_{x\neq y}\rho_y\abs{\mathcal K_{xy}}+\abs{1-\alpha\rho_y\mathcal K_{xy}}\nonumber\\&&+\gamma\alpha\rho_y\mathcal K_{xy}\Big).\nonumber\\
&=&\max_x\Big(1+(1+\gamma)\sum_{x\neq y}\rho_y\abs{\mathcal K_{xy}}-(1-\gamma)\rho_x\mathcal K_{xx}\Big).\nonumber
\end{eqnarray}
Observe that $G(\mathcal K)<1$ if an only if
\[\forall x,\;(1-\gamma)\rho_x\mathcal K_{xx}\geq(1+\gamma)\sum_{x\neq y}\rho_y\abs{\mathcal K_{xy}}. \] See~\cite{achiam2019towards} for other details.
\end{proof}
Considering $\mathcal U^{iii}$, we observe that the conditions in  \textbf{Theorem~\ref{thm:1}} are quite confining, it not only requires $\rho>0$ everywhere but also for  the choices of $\gamma$ (eg. $\gamma=0.999$) the (diagonal terms of $\mathcal K$)$\gg$(off-diagonal terms of $\mathcal K$) . This observation provides us another important design insight.

{\sc Design Insight 3}. The stability of DQL-MPTCP depends on the properties of the Q-approximator; in the sense that approximators which are more aggressive (larger off-diagonal terms in $\mathcal K$) may not demonstrate stable learning.

Generally, in our DQL-MPTCP learning, both the $\mathcal K$ and the $\rho$ change between update stages, thus, each stage can be perceived as applying different updates. On this end, in general, if we sequentially apply different contraction maps with the same fixed point, then we will atain that fixed point. As a result, motivated by the findings in~\cite{achiam2019towards}, we have the following theorem.

\begin{theorem}
Assume a sequence of updates \{$\mathcal U^{0}, \mathcal U^{1}, \mathcal U^{2}$, \dots\} with each $\mathcal U^{i}: \mathcal Q\rightarrow \mathcal Q$ being Lipschitz continuous with constant $\delta_i$ and the existence and uniqueness of the fixed point $Q^{\star}=\mathcal T^{\star} Q^{\star}$ is guaranteed, where all $\mathcal U^{i}$ share a common fixed point $Q^{\star}$. Then, starting from any initial point ${Q^{0}}$, the trajectory \{{$Q^{0}, {Q^{1}}, {Q^{2}}, \dots$}\} generated by the DQL-MPTCP algorithm converges to a unique and globally stable fixed point $Q^{\star}$, if and only if the stages produced by $Q_{i+1}=\mathcal U^{i} Q_{i}$ satisfies ({for an iterate $k$}):
\begin{equation}\tag{A11.4}
    \norm{Q^{\star}-Q^0}\prod_{j=0}^{i-1}\delta_j\geq\norm{Q^{\star}-Q^i}\;\mbox{such that}\; \forall j\geq k, \delta_j\in[0,1).
\end{equation}
\end{theorem}
\begin{proof}
Using fixed-point assumption and iterative sequence of updates,
$\norm{Q^{\star}-Q^i}$ = $\norm{\mathcal U^{i-1}Q^{\star}-\mathcal U^{i-1}Q^{i-1}}$.
From the definition (and property) of Lipschitz continuity  $\norm{Q^{\star}-Q^i} \leq \delta_{i-1}\norm{Q^{\star}-Q^{i-1}}$. Therefore, $\norm{Q^{\star}-Q^0}\prod_{j=0}^{i-1}\delta_j\geq\norm{Q^{\star}-Q^i}.$
Finally,  sequence \{{$Q^{0}, {Q^{1}},\dots$}\} for $\forall j\geq k,\; \delta_j\in[0,1)$ converges to $Q^{\star}$ as $\lim_{m\to\infty}\prod_{j=k}^m\delta_j=0.$
\end{proof}
\begin{table*}[!h]
\centering
\caption{Metrics Used For Network Emulation Experiments}
\begin{tabular}{c|l|l}
  \hline
 \textbf{Experiment ID}& \textbf{Metrics used in the emulated Network Settings and Scenario}  & \textbf{Remarks}\\ \hline\hline
  I & Moving Average Throughputs: Arithmetic mean of a set of previous& Compare DQL-MPTCP with\\
   & throughputs until $t$ seconds& LIA, OLIA, BALIA\\\hline
  II & Throughputs vs. Fluctuating Delay: Variation in throughputs with& How DQL-MPTCP reacts to\\
   &   increasing one path delays& fluctuating delays in one of the path\\\hline
  III & Throughputs vs. Varying Bandwidth: Average throughputs with& How DQL-MPTCP performs\\
   &  increasing one path capacities& with respect to varying capacities\\\hline
 IV & TCP Friendliness : Ability of a new protocol to behave under& Degradation of DQL-MPTCP\\
   &congestion like the TCP protocol& when TCP competes for bandwidth\\\hline
\end{tabular}
\label{defination2}
\end{table*}

\subsection{Stability by Design}
\label{sec:inst}
Even though DQL-MPTCP updates vary between stages, the aforementioned
insights from different settings provide useful guidance
for understanding and addressing possible  instability issues in DQL-MPTCP.
With the relevant intuitions from the aforementioned analysis,
we next discuss four different possible causes of instability and how we address them by design.

{\sc Cause 1.} Consider a scenario with aggressive learning rate ($\alpha$ is very high), then the term $\BigO{\norm {g'-g}^2}$ in (A9) is quite large, as a result DQL-MPTCP updates may not correlate well with the Bellman updates thus leading to instability.

{\sc Cause 2.} A scenario when $\alpha$ is small enough for linearization, however, quite large that the $\mathcal U$ in (A11.3) start expanding rather then contracting, hence, leading to instability (Theorem~\ref{thm:3}).

{\sc Cause 3.} Too aggressive generalization of Q function due to large off-diagonal matrix of $\mathcal K$ also may cause $\mathcal U$  to expand (Theorem~\ref{thm:3}).

{\sc Cause 4.} If  the distribution for updates are inadequate, the
Q-values for missing state action pairs which are computed
by general extrapolation, often incur errors. Such errors, which 
propagate through the Q-values to all other state-action pairs, may
lead to unstable learning (Theorem~\ref{thm:2}).

Observe that the three insights from Section~\ref{sec:stab1} have been exploited for the convergence and stability in the DQL-MPTCP design (Section~\ref{sec:algo12}). Earlier in Section~\ref{sec:algo12}, we adopted smooth learning rates in relation to the Bellman updates for guaranteed contraction and linearization. We have also applied constructive generalization of the Q function and utilized learning from demonstrations to tackle the case of missing values.

\section{Performance Evaluation}

We have executed an extensive set of experiments to evaluate DQL-MPTCP under several network scenarios. 
The network settings used in our testbed are discussed next. Thereafter, we present our findings and reflections from the obtained results.

 We evaluate and compare the performance of our DQL-MPTCP algorithm with the standard MPTCP algorithms whose implementation code is accessible: LIA~\cite{raiciu2011coupled}, BALIA~\cite{peng2016multipath}, OLIA~\cite{khalili2013mptcp}. We have used MPTCP v0.93 implemented in Linux.\footnote{ \emph{https://www.multipath-tcp.org}} We explain the findings from four different experiments in the following.  Table ~\ref{defination2} describes the performance metrics under consideration for the four experiments. 
 
 Our testbed consists of five laptops, four of them are users and one acts as a Server (running MPTCP Source).    All four user laptops are connected  with  a  Switch (Gigabit) consisting of two  separate interfaces (Gigabit Ethernet) creating two different connections, always used for downloading files from the server laptop. For flexibility, quite similar to ~\cite{peng2016multipath}, every MPTCP  connection  handles  two different subflows in our experiments. We employed and enhanced  Linux traffic control facilities\footnote{NetEm:{https://wiki.linuxfoundation.org/networking/netem}} and   conduct experiments in the controlled environment not only by varying bandwidth, but also under different asymmetric path loss and delays.  Each of our experiments consists of eight MPTCP connections (with two subflows)  connecting  the  server with  users.  We analyze the results based on several data traffic dynamics upon downloading files of variable size (3MB to 500MB) from the  server. The aggregate throughputs are extracted from tcpdumps, averaged approximately 900 different experiment runs and numbers in the figures represent the average for a connection (unless stated otherwise).

\subsection{Results and Findings}

Our experiments and findings are discussed as follows. Table ~\ref{defination2} describes the organization of four different experiments.

\subsubsection{Experiment I} Our testbed is set  at equal delay of $50$ ms,  equal bandwidths $20$~Mbps and  the same 3\% packet loss probability to quantify the impact on moving-average throughputs for $60$ seconds  upon downloads. We use the same file of size ($30$MB) for all MPTCP algorithms (several experimental runs with \emph{DQL, BALIA, OLIA, LIA} separately using the same setup). Figure~\ref{fig:short} shows the effect on the running average throughputs.

 \begin{figure}[t]
\centering 
\includegraphics[width=3.025 in, height=1.85145 in]{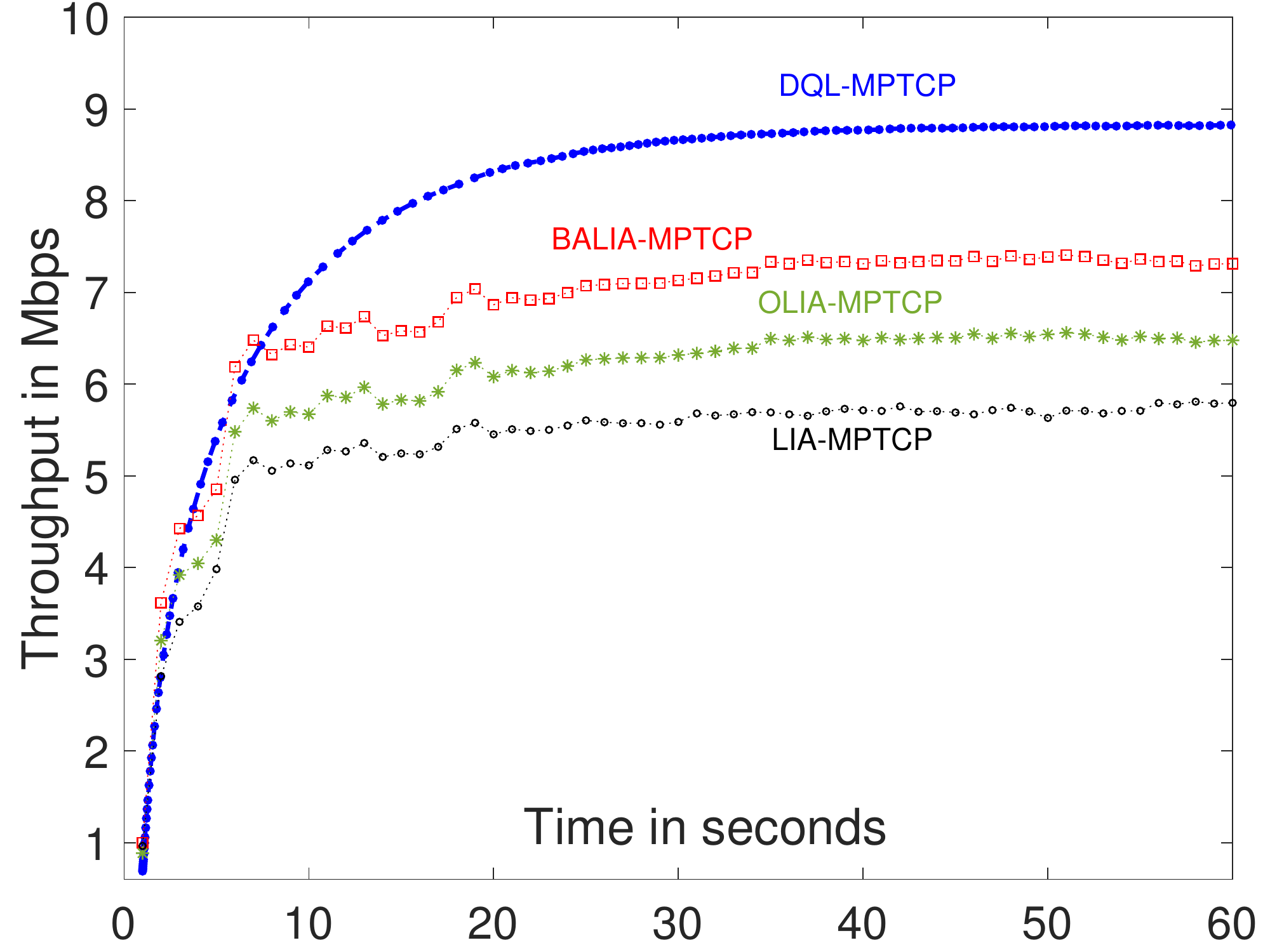}
\caption{\rm Experiment I: For both paths, we set equal delay ($50$ms), equal $3\%$ path loss and equal bandwidth ($10$~Mbps each) while downloading replicas of $600$MB  file. We perform different runs for \emph{DQL, BALIA, OLIA, LIA} with same setup. \label{fig:short}} 
\end{figure}

 {\sc{Finding 1:}} In terms of perceived throughputs, for small size file transfers and short flows, our DQL-MPTCP outperforms the standard MPTCP algorithms under realistic channel conditions. 

\begin{remark}
{\sc Finding 1} is the joint effect of coupled congestion control and scheduling policy adopted in this    study. Note that the standard algorithms, LIA, OLIA, and BALIA inherent separately designed congestion control and scheduling.
\end{remark}
\subsubsection{Experiment II}
For both paths, we set equal bandwidth (20 Mbps) and equal  packet loss probabilities of 5\%  to see the impact on the aggregate throughput. One path delay increases from 20 ms to 100 ms after every 20 seconds, but the delay of another path is fixed at 20 ms.

Note that the users are downloading the same file of  size $600$ MB but using different MPTCP algorithms (different experimental runs with different versions of MPTCP:  \emph{DQL, BALIA, OLIA, LIA}) while the network setup is exactly the same. Figure~\ref{fig:delay} depicts the aggregate impact on throughputs by varying the delay of one of the path dynamically, which provides us another interesting finding.  
  
  \begin{figure}[t]
\centering 
\includegraphics[width=3.02 in, height=1.85145 in]{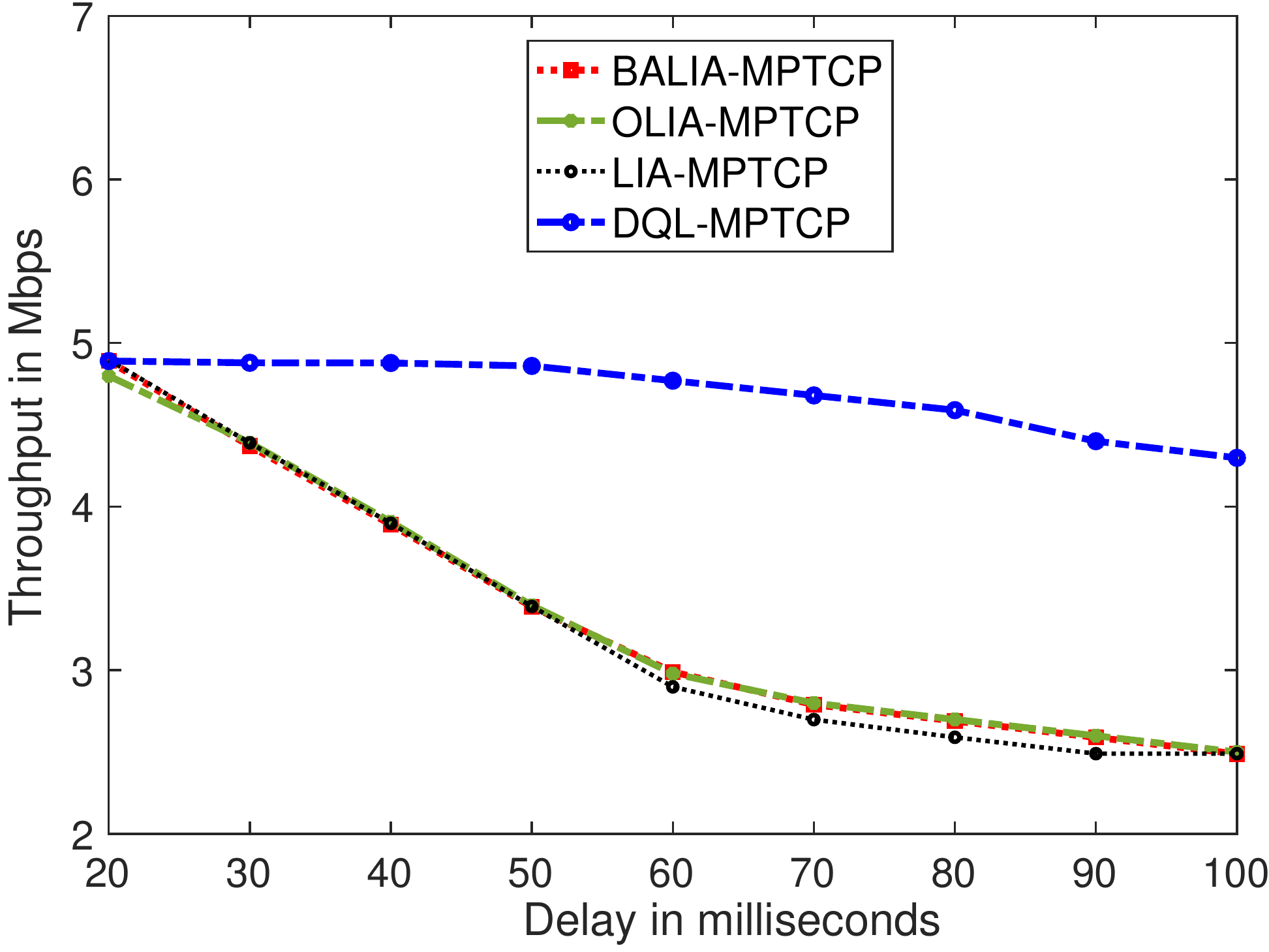}
\caption{\rm Experiment II: Change the delay of one path gradually after every 20 seconds (increasing by 10~ms) starting from 20~ms to 100~ms;  while downloading a large file. Delay over another path is set constant (20~ms). Both paths are set at same bandwidths ($10$~Mbps each) and suffers from the same path loss of $5\%$. 
\label{fig:delay}} 
\end{figure}
  
{\sc{Finding 2:}} By increasing one of the path delay while downloading, the aggregate throughputs obtained by all MPTCP algorithms decreases. This behavior is as expected, but, the decreasing rate of throughputs is quite dissimilar with different MPTCP algorithms. In particular, the slope is quite high with all other MPTCP algorithms but is smallest with DQL-MPTCP.

\begin{remark}
\emph {\sc{Finding 2}} is due to the delayed feedback of standard MPTCP algorithms and time to discover and adapt with the temporal changes in path delay. As a consequence, most  MPTCP algorithms need retransmission and recovery of lost packets. This may also require re-sequencing of received packets at the user end  (all packets need to be delivered in order). However, DQL-MPTCP is robust to delay variations as the DQL agent is learning packet scheduling and congestion control continuously. Further, the improved responsiveness is also due to the use of prioritized experience replay buffer, which accelerates learning and adapts packet rates and schedule across paths quickly. As a result, with these time-varying delay changes, the reordering delay and waiting time for packets transmitted over the slow path, compared to standard MPTCPs, is minimal with DQL-MPTCP. 
\end{remark}

\subsubsection {Experiment III}
For both paths, we fixed the same delays of $150$ ms and $4\%$ packet loss probability, equal for both paths. Then, we vary bandwidths dynamically for only one  of the paths (other path bandwidth is constant 20~Mbps) to see the impact on the average throughputs.  All MPTCPs are downloading same size ($600MB$) file over  the same testbed setup. Figure~\ref{fig:bandwidth} illustrates the impact on throughputs.\\
\begin{figure}[t]
\centering 
\includegraphics[width=3.025 in, height=1.85145 in]{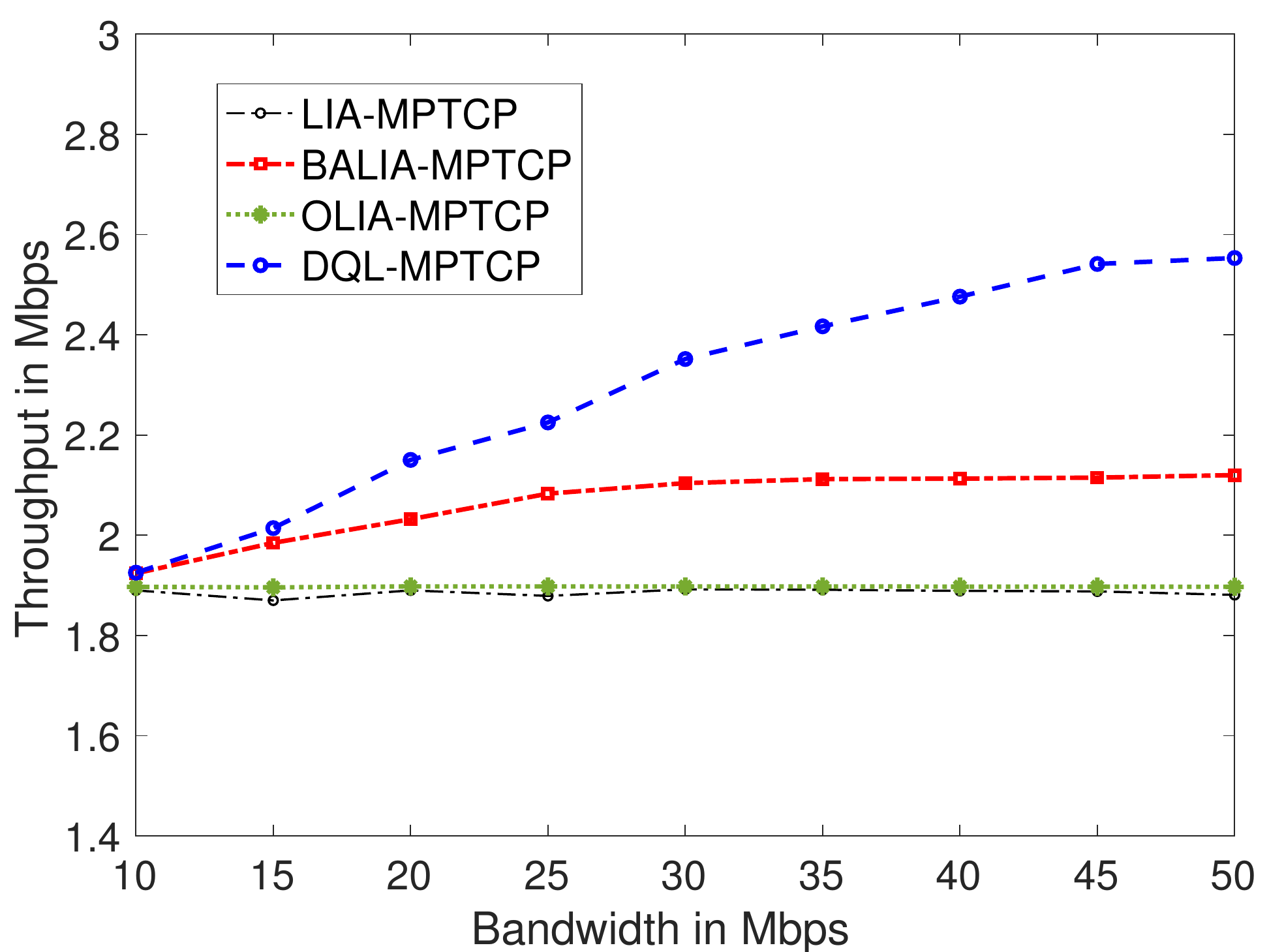}
\caption{\rm Experiment III: We  vary  bandwidths of one path from 10~Mbps to 50~Mbps dynamically fixed another path at 10~Mbps). We set equal delay (150~ms) and same path loss ($4\%$) for both paths to observe the impact of time-varying capacities on performance as perceived by all MPTCP algorithms. \label{fig:bandwidth}} 
\end{figure}
{\sc{Finding 3}:} {DQL-MPTCP} obtains the highest throughput. All other state-of-the-art standard MPTCP algorithms fall behind in utilizing dynamically varying bandwidths.\\ 
Figure~\ref{fig:bandwidth} shows that throughputs obtained by standard MPTCP algorithms are fairly low. DQL-MPTCP algorithm outperforms others in the same testbed setting. This is due to the fact that under realistic channel error, the standard MPTCP algorithms are guided by a predefined fixed loss and delay based policy.

\begin{remark}
With the enhanced continuous learning using optimal packet scheduling action of DQL-MPTCP connections, they utilize increasing bandwidth pretty well even under significant magnitude of path loss. More importantly, the underlying LSTM inside the DQL agent of our MPTCP can quickly discover network bandwidth changes and intelligently perform balancing with joint packet scheduling and congestion control. Design of MPTCP algorithm requires balancing the
tradeoff between fairness to single path TCP, when sharing a bottleneck,  and responsiveness in the sense that MPTCP grabs available bandwidth when it becomes available in any of the available subflows~\cite{raiciu2011coupled, peng2016multipath}. 
\end{remark}

As noted in earlier experiments, we observed considerably high responsiveness of DQL-MPTCP. Our next experiment is to study the fairness of DQL-MPTCP to single path TCP.
\subsubsection{Experiment IV}
 Figure~\ref{fig:fig8} illustrates the throughputs comparisons for the TCP and MPTCP connections, when we set single path TCP flows competing with MPTCP connection in one of the paths. 
\begin{figure}[t]
\centering 
\includegraphics[width=3.025 in, height=1.85145 in]{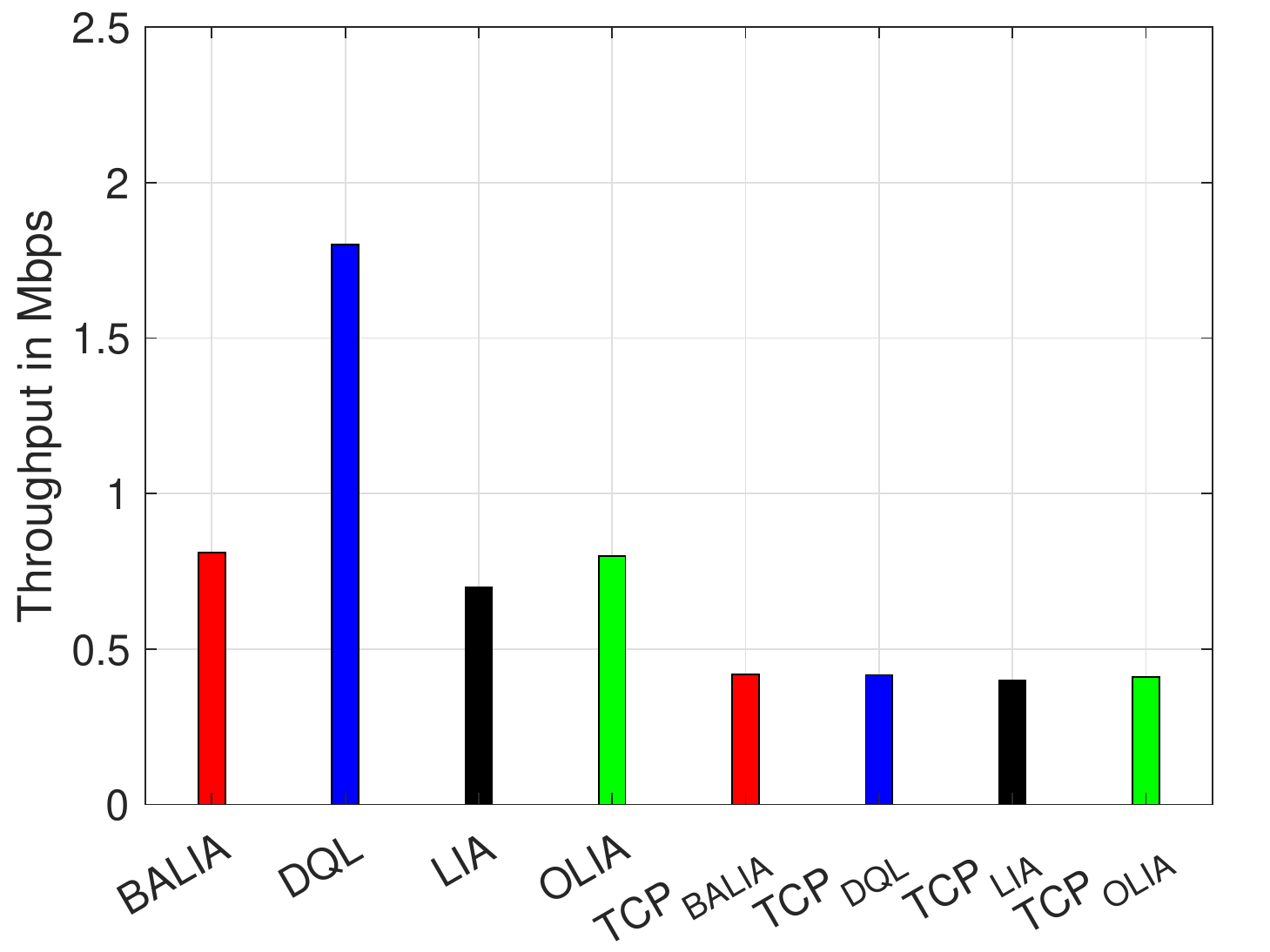}
\caption{\rm Experiment IV setting: MPTCP flows contending with single path TCP. Using the setup of Experiment III, we  conduct experiment with five MPTCP connections  sharing one of the path with rival five single path TCP connections. We repeat same experimental runs using all four different versions of MPTCP: DQL, BALIA, OLIA and LIA algorithms.\label{fig:fig8}} 
\end{figure}\\
 Observe in Figure~\ref{fig:fig8} that DQL-MPTCP is fair to rival single path TCP connections in terms of throughputs.  It shows that the throughput perceived by single path TCP is quite close to that when competing with standard MPTCP algorithms; however, the DQL-MPTCP connection achieves higher throughputs. As compared to BALIA vs single path TCP, the throughput obtained by the single path TCP when competing with DQL-MPTCP subflows is negligibly smaller. The reason for this observation (negligibly smaller single path TCP throughput when competing with DQL-MPTCP) requires further  investigations.
 
\section{ \textcolor{black}{Future Directions}}
 \textcolor{black}{With relevant insights from this work, we have identified the following three research problems for future directions.}
 \subsection{\textcolor{black}{Cost assessment Study}}
 \textcolor{black}{The cost assessment study of DQL-MPTCP with relevant insights from this work is another exciting research direction. At a higher layer, we can design a high-level policy framework i) to consider efficient energy consumption and economic data plans (e.g. Cellular vs. WiFi or number of interfaces in 5G and beyond) ii) to determine a set of suitable paths based on cost/energy constraints. Such high-level policy engine could determine the (cost/energy) optimal set of paths.}
\subsection{\textcolor{black}{A hybrid (model- and DQL-based) design}}
 \begin{figure}[b]
\centering 
\includegraphics[width=3.025 in]{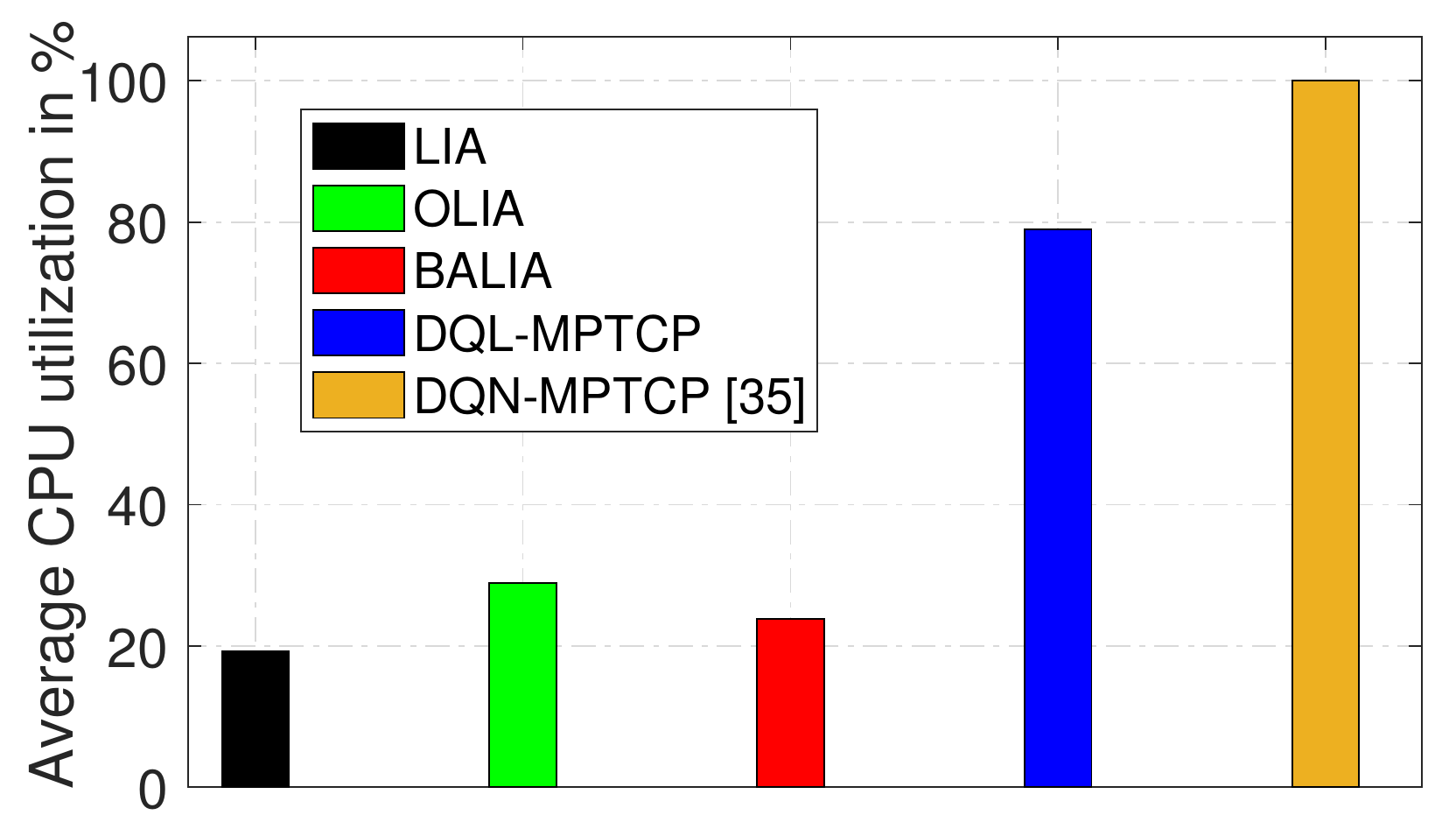}
\caption{\rm   \textcolor{black}{We send traffic from a server to the client (with two 10 Mbps links and 20ms delay) for 200 seconds. We measure the average CPU use on the MPTCP source side to examine learning overhead and compare DQL-MPTCP with other models.
We repeat the same experiment and record iPerf's CPU utilization. We repeat same experimental runs and measure the average \% utilization of the five different versions of MPTCP~\cite{pokhrel2020multipath}.}\label{fig:fig9}} 
\vspace{-4 mm}
\end{figure}
\textcolor{black}{Model-based designs such as LIA, OLIA, BALIA typically have very low computing overhead,
as they adopt a fine-grain scale for their control
loops and respond to every ack. In sharp contrast,
DQL-based designs have potentially high computing
overheads as their feedback loops are longer than
the model-based designs.}

\textcolor{black}{To investigate the overhead cost of DQL-MPTCP, we send traffic from a server to the client (with two 10 Mbps links and 20ms delay) for 200 seconds. We measure the average CPU use on the MPTCP source side to examine learning overhead and compare DQL-MPTCP with other models.
We repeat the same experiment and record iPerf's CPU utilization.}

\textcolor{black}{We remove the first few seconds of the experiment for all designs in order to have a reasonable comparison and limit the effect of the startup phases of different MPTCP designs.
Figure~\ref{fig:fig9} shows results. DQL-MPTCP, as expected, has a high overhead compared to other model-based designs (but lower than the DQN-MPTCP~\cite{pokhrel2020multipath}). Although the current DQL-MPTCP user-space implementation has a high overhead, we think the final optimized kernel version would lower overhead.} 

 \textcolor{black}{This experiment provides new insights and an essential future direction towards a pragmatic and evolutionary strategy of rethinking a hybrid MPTCP (by blending model-based techniques with advanced DQL-based approaches).}
 
 \subsection{\textcolor{black}{Minimizing Training overhead}}
 \textcolor{black}{We trained the DQL agent for over twenty-five thousand transition samples using iPerf3 (\emph{https://ipref.fr}), which generates packets continuously to keep the paths always busy in our network settings. This is useful to generate sufficient training samples. We perform experiments with different settings under variable packet loss, delays and bandwidths, hence our (offline) training time varies (2 to 10 hours). In one of the settings with two paths, the training runs for $3$ hours when we set equal bandwidths $ 10$Mbps each, equal RTTs $550$ms  each and equal packet loss probabilities  $ 0.5$\% for each path. Nevertheless, the outcome is beneficial: the offline training makes our algorithm fully prepared for online deployment by minimizing the startup  delay.}

\textcolor{black}{It is worth noting that offline training is a one time job. Recall that we have used Deep Neural Networks (DNNs) for inference in our DQL-MPTCP implementation, each of which owns only two (hidden) layers. The online inference time observed during our experimentation is just about 0.8ms (which causes smallish overhead for real-time decision process). Retraining our DQL agent is needed only when the network settings changes drastically (e.g., from \emph{high-bandwidth-delay} to \emph{low-bandwidth-delay}). The cause for this is to collect sufficient transition samples needed  to update (and learn) the DNNs and gain sufficient experience to make better decisions when comparable conditions occur in a new network environment. An effective way to conduct retraining of a trained DQL agent to adapt for the new network services, such as ultra reliable low latency communication, will benefits from transfer learning~\cite{pokhrel2021multipath} and needs further investigations~\cite{pokhrel2021urllc}.}

\section{Conclusions}
We developed a novel DQL-based multipath congestion control and packet scheduling scheme using policy gradients and prioritized replay buffer for the future Internet and wireless technologies. It utilizes a DQL agent with policy gradients to jointly perform dynamic packet scheduling and congestion control, and conduct network reward maximization for all active subflows of an MPTCP connection. We provided stability analysis of DQL-MPTCP with relevant practical design insights. Our design consists of training intelligent agent for interface and actuation, which possesses a flexible LSTM-based representation network. The developed architecture appears to be capable of i) learning an effective representation of all MPTCP subflows, ii) jointly performing intelligible coupled congestion control and packet scheduling, and iii) dealing with dynamic network characteristics and time-varying flows. 

 We observed better performance of DQL-MPTCP when comparing with standard MPTCP algorithms under similar network settings. Essentially, what the DQL policy may amount to is that each MPTCP agent either sends or abstains from sending based on a set of observables and the policy attempts to maximize goodput and maintain low delays over the whole network and for all contending agents. Moreover, we anticipate that the DQL should have accelerated MPTCP in discovering the best policy for the underlying partially observable Markov process (of the network dynamics), which requires further investigations and is left for future work.

\bibliographystyle{IEEEtran}
\bibliography{shiva}
\end{document}